\documentclass[10pt, journal]{IEEEtran}

\usepackage{graphics} 
\usepackage{amsmath} 
\usepackage{amssymb}  
\usepackage{bm}
\usepackage{subfigure}
\usepackage{acronym}
\usepackage{hyperref}
\usepackage{paralist}
\usepackage{float}
\usepackage{color}
\usepackage{multicol}
\usepackage{tikz}

\makeatletter
\hypersetup{colorlinks=true}
\AtBeginDocument{\@ifpackageloaded{hyperref}
  {\def\@linkcolor{blue}
  \def\@anchorcolor{red}
  \def\@citecolor{red}
  \def\@filecolor{red}
  \def\@urlcolor{black}
  \def\@menucolor{red}
  \def\@pagecolor{red}
\begingroup
  \@makeother\`%
  \@makeother\=%
  \edef\x{%
    \edef\noexpand\x{%
      \endgroup
      \noexpand\toks@{%
        \catcode 96=\noexpand\the\catcode`\noexpand\`\relax
        \catcode 61=\noexpand\the\catcode`\noexpand\=\relax
      }%
    }%
    \noexpand\x
  }%
\x
\@makeother\`
\@makeother\=
}{}}
\makeatother

\usepackage{amsthm}

\DeclareMathOperator{\sign}{sign}

\newtheorem{Theorem}{Theorem}

\newtheorem{Lemma}{Lemma}

\newtheorem{Remark}{Remark}

\newtheorem{Cor}{Corollary}

\newtheorem{Assumption}{Assumption}

\newtheorem{Definition}{Definition}

\def\BibTeX{{\rm B\kern-.05em{\sc i\kern-.025em b}\kern-.08em
    T\kern-.1667em\lower.7ex\hbox{E}\kern-.125emX}}

\begin{document}
\title{Finite-Time Stability of Switched and Hybrid Systems with Unstable Modes}

\author{Kunal Garg, \IEEEmembership{Student Member, IEEE} and Dimitra Panagou, \IEEEmembership{Senior Member, IEEE}
\thanks{
The authors would like to acknowledge the support of the Air Force Office of Scientific Research under award number FA9550-17-1-0284.}
\thanks{The authors are with the Department of Aerospace Engineering, University of Michigan, Ann Arbor, MI, USA; \texttt{\{kgarg, dpanagou\}@umich.edu}.}
}

\maketitle

\begin{abstract}
In this work, we study finite-time stability of switched and hybrid systems in the presence of unstable modes. We present sufficient conditions in terms of multiple Lyapunov functions for the origin of the system to be finite time stable. More specifically, we show that even if the value of the Lyapunov function increases in between two switches, i.e., if there are unstable modes in the system, finite-time stability can still be guaranteed if the finite time convergent mode is active long enough. In contrast to earlier work where the Lyapunov functions are required to be decreasing during the continuous flows and non-increasing at the discrete jumps, we allow the Lyapunov functions to increase \emph{both} during the continuous flows and the discrete jumps. As thus, the derived stability results are less conservative compared to the earlier results in the related literature, and in effect allow the hybrid system to have unstable modes. 
Then, we illustrate how the proposed finite-time stability conditions specialize for a class of switched systems, and present a method on the synthesis of a finite-time stabilizing switching signal for switched linear systems. As a case study, we design a finite-time stable output feedback controller for a linear switched system, in which only one of the modes is controllable and observable. Numerical example demonstrates the efficacy of the proposed methods. 
\end{abstract}

\begin{IEEEkeywords}
Finite-Time Stability; Hybrid Systems;  Multiple Lyapunov Functions.  
\end{IEEEkeywords}

\section{Introduction}
Many real-world systems exhibit properties of continuous evolution and discrete jumps at times, which are termed as \textit{hybrid systems}. 
Hybrid systems are capable of modeling large class of complex dynamical systems. The introductory paper \cite{antsaklis1998hybrid} provides an overview and the merits of using hierarchical organization within a hybrid systems framework; namely, that it helps in managing complexity since it requires less detailed models at higher levels. The class of \textit{switched systems} that includes the variable structure system and the multi-modal system is an important subcategory of hybrid systems. There is a variety of practical examples where certain stability properties cannot be achieved using a single continuous feedback, and as thus a switching controller becomes essential; for instance, the authors in \cite{zhao2001hybrid} make their case for the well-studied pendulum on a cart problem. Many control theoretic examples have been proposed where switched controller system can provide stability and performance guarantees; see e.g. \cite{ishii2002stabilizing,narendra2003adaptive,savkin1999robust}. Stability of hybrid systems has been studied extensively in the literature; for an overview of the theory of switched and hybrid systems, i.e., on solution concepts, notion of stability, the interested readers are referred to \cite{liberzon2003switching} and \cite{lygeros2004lecture,goebel2012hybrid}, respectively.

\subsection{Stability of switched systems}
Stability of switched system has been analyzed by many researchers in the past. The survey articles \cite{lin2009stability,davrazos2001review} and the references therein give a detailed overview of various stability results for switched and hybrid systems. Stability of switched systems is typically studied using either a \textit{common} Lyapunov function, or \textit{multiple} Lyapunov functions. The book \cite{liberzon2003switching} discusses necessity and sufficiency of the existence of a common Lyapunov function for all subsystems of a switched system for asymptotic stability under arbitrary switching. The authors in \cite{ishii2002stabilizing} study linear switched systems with dwell-time using a common quadratic control Lyapunov function (CQLF) and state-space partitioning. In the review article \cite{shorten2007stability}, the authors study the stability of switched linear systems and linear differential inclusions. They present sufficient conditions for the existence of CQLFs and discuss converse Lyapunov results for switched systems. In \cite{branicky1998multiple}, the author introduces the concept of multiple Lyapunov functions to analyze stability of switched systems; since then, a lot of work has been done on the stability of switched systems using multiple Lyapunov functions \cite{zhao2008stability,zhao2012stability,lin2009stability}. In \cite{zhao2008stability}, the authors relax the non-increasing condition on the Lyapunov functions by introducing the notion of generalized Lyapunov functions. They present necessary and sufficient conditions for stability of switched systems under arbitrary switching. In \cite{zhao2012stability}, the authors introduce the concept of Multiple Linear Copositive Lyapunov functions (ML-CLFs) and give sufficient conditions for exponential stability of Switched Positive Linear Systems (SPLS) in terms of feasibility of Linear Matrix Inequalities (LMI). 
The authors in \cite{zhao2017new} use discontinuous multiple Lyapunov functions in order to guarantee stability of \textit{slowly} switched systems, where the stable subsystems are required to switch slower (i.e., stay active for a longer duration) as compared to unstable subsystems.

\subsection{Stability of hybrid systems}
Unlike switched systems, where only the dynamics of the system is allowed to have jumps, the notion of hybrid systems is more general and allows the system states to have discrete jumps as well. The survey paper \cite{teel2014stability} studies Lyapunov stability (LS), Lagrange stability and asymptotic stability (AS) for stochastic hybrid systems (SHS), and provides Lyapunov conditions for stability in probability. The paper also presents open problems on converse results on the stability in probability of SHS. 
In \cite{liu2016lyapunov}, the authors study hybrid systems exhibiting delay phenomena (i.e., memory). They establish sufficient conditions for AS using Lyapunov-Razumikhin functions and Lyapunov-Krasovskii functionals. More recently, pointwise AS of hybrid systems is studied in \cite{goebel2016notions}, where the notion of set-valued Lyapunov functions is used to establish sufficient conditions for AS of a closed set. In \cite{teel2013lyapunov}, the authors impose an average dwell-time for the discrete jumps and devise Lyapunov-based sufficient conditions for exponential stability of closed sets.

\subsection{Related work on FTS}
In contrast to AS, which pertains to convergence as time goes to infinity, finite-time Stability (FTS)\footnote{With slight abuse of notation, we use FTS to denote the phrase "finite-time stability" or "finite-time stable", depending on the context.} is a concept that requires convergence of solutions in finite time. FTS is a well-studied concept, motivated in part from a practical viewpoint due to properties such as convergence in finite time, as well as robustness with respect to disturbances \cite{ryan1991finite}. In the seminal work \cite{bhat2000finite}, the authors introduce necessary and sufficient conditions in terms of Lyapunov function for continuous, autonomous systems to exhibit FTS, with focus on continuous-time autonomous systems. 

FTS of switched/hybrid systems has gained popularity in the last few years. The authors in \cite{liu2017finite} consider the problem of designing a controller for a linear switched system under delay and external disturbance with finite- and fixed-time convergence. In \cite{li2013robust}, the authors design a hybrid observer and show finite-time convergence in the presence of unknown, constant bias.  In \cite{nersesov2007finite}, the authors study FTS of nonlinear impulsive dynamical systems, and present sufficient conditions
to guarantee FTS. The work in \cite{li2013robust,nersesov2007finite} considers discrete jumps in the system states in a continuously evolving system, i.e., one model for the continuous dynamics, and one model for the discrete dynamics. The authors in \cite{li2019finite} present conditions in terms of a common Lyapunov function for FTS of hybrid systems. They require the value of the Lyapunov function to be decreasing during the continuous flow and non-increasing at the discrete jumps. 

The authors in \cite{rios2015state} design an FTS state-observer for switched systems via a sliding-mode technique. In \cite{orlov2004finite}, the authors introduce the concept of a \textit{locally} homogeneous system, and show FTS of switched systems with uniformly bounded uncertainties. More recently, \cite{zhang2018finite} studies FTS of homogeneous switched systems by introducing the concept of hybrid homogeneous degree, and relating negative homogeneity with FTS. In \cite{fu2015global}, the authors consider systems in strict-feedback form with positive powers and design a controller as well as a switching law so that the closed-loop system is FTS. The authors in \cite{li2016results} present conditions in terms of a common Lyapunov function for FTS of hybrid systems. They require the value of the Lyapunov function to be decreasing during the continuous flow, and non-increasing at the discrete jumps. In \cite{bejarano2011finite}, the authors design an FTS observer for switched systems with unknown inputs. They assume that each linear subsystem is \textit{strongly} observable, and that the first switching occurs after an \emph{a priori} known time. In contrast, in the current paper we do not assume that the subsystems are homogeneous or in strict feedback form, and present conditions in terms of multiple Lyapunov functions for FTS of the origin.

The work in  \cite{amato2008suff} studies FTS of impulsive dynamical linear systems (IDLSs); impulsive systems describe the evolution of systems where the continuous development of a process is interrupted by abrupt changes of
state \cite{bonotto2016survey}. In \cite{amato2013necessary}, the authors extend these results and show that the conditions in \cite{amato2008suff} are also necessary for FTS of IDLs. The authors in \cite{li2016results} present conditions in terms of a common Lyapunov function for establishing FTS of hybrid systems. In \cite{zhangfinite}, the authors consider the switched system with an assumption that each subsystem possess a homogeneous Lyapunov function and that the switching-intervals are constant. 

\subsection{Our Contributions}
In this paper, we consider a general class of hybrid systems, and develop sufficient conditions for FTS of the origin of the hybrid system in terms of multiple Lyapunov functions. \textit{To the best of authors' knowledge, this is the first work considering FTS of switched or hybrid systems using multiple Lyapunov functions}. The main contributions are summarized as follows. 

\textbf{FTS of hybrid systems with unstable modes}: We first define the notion of FTS for hybrid systems so that it does not restrict each mode of the hybrid system to be FTS in itself. More specifically, we relax the requirement in \cite{zhao2008stability,li2019finite,li2016results} that the Lyapunov function is non-increasing at the discrete jumps, and strictly decreasing during the continuous flow; instead, we allow the multiple Lyapunov functions to increase both during the continuous flow and at the discrete jumps, and only require that these increments are bounded. In this respect, we allow the hybrid system to have unstable modes while still guaranteeing FTS. In addition, we present a novel proof on the  stability of the origin using multiple Lyapunov functions under the aforementioned relaxed conditions.
In contrast to \cite{li2013robust,nersesov2007finite}, we consider the general case with $N_f$ continuous flows and $N_g$ discrete jump dynamics, where $N_f,N_g$ can be any positive integers. The main result is that if the origin is \textit{uniformly} stable, i.e., stable under arbitrary switching, and if there exists an FTS mode that is active for a sufficient cumulative time, then the origin of the resulting hybrid system is FTS. 

\textbf{FTS of switched systems}: Then, we demonstrate how the results specialize for a class of switched systems with unstable subsystems. In contrast to \cite{rios2015state,fu2015global}, we do not assume that the subsystems of the switched system are homogeneous or in strict feedback form, and present conditions for FTS of a class of switched systems. 

\textbf{Switching-signal design for FTS, and applications}: We present a method for designing a switching signal so that the origin of the resulting switched system is FTS. We then apply our developed methods to design an FTS output controller for switched linear systems for the case when only one of the subsystems (or modes) is controllable and observable.

\subsection{Organization}
The paper is organized as follows: In Section \ref{FTS HS}, we present an overview of FTS followed by conditions for FTS of hybrid systems. In Theorem \ref{FT Th 2}, we present conditions for LS and FTS of the origin in terms of multiple Lyapunov functions, and then, show that uniform stability and sufficiently long activation of FTS mode is sufficient for FTS in Corollary \ref{cor LS FTS mode}. In Section \ref{sec Switch System}, we specialize our results for a class of switched systems. We also present a method of designing finite-time stabilizing switching law and as a case study, design FTS output-feedback for switched linear system for the case when only one of the subsystems is both controllable and observable. Section \ref{Simulations} evaluates the performance of the proposed method via simulation results. Our conclusions and thoughts on future work are summarized in Section \ref{Conclusions}.

\section{FTS of Hybrid Systems}\label{FTS HS}

\subsection{Preliminaries}

We denote by $\|\cdot\|$ the Euclidean norm of vector $(\cdot)$, $|\cdot|$ the absolute value if $(\cdot)$ is scalar and the length if $(\cdot)$ is a time interval. The set of non-negative reals is denoted by $\mathbb R_+$, set of non-negative integers by $\mathbb Z_+$ and set of positive integers by $\mathbb N$. We denote by $\textrm{int}(S)$ the interior of the set $S$, and by $t^-$ and $t^+$ the time just before and after the time instant $t$, respectively. 

\begin{Definition}\label{class GK}
(\textbf{Class-$\mathcal{GK}$ function}): A function $\alpha:D\rightarrow\mathbb R_+$, $D\subset \mathbb R_+$, is called a class-$\mathcal{GK}$ function if it is increasing, i.e., for all $x>y\geq 0$, $\alpha(x)>\alpha(y)$, and right continuous at the origin with $\alpha(0) = 0$.
\end{Definition}

\begin{Definition}\label{class GK infty}
(\textbf{Class-$\mathcal{GK}_\infty$ function}): A function $\alpha:\mathbb R_+\rightarrow\mathbb R_+$ is called a class-$\mathcal{GK}_\infty$ function if it is a class-$\mathcal{GK}$ function, and $\lim_{r\to\infty}\alpha(r) = \infty$. 
\end{Definition}
Note that the class-$\mathcal{GK}$ (respectively, $\mathcal{GK}_\infty$) functions have similar composition properties as those of class-$\mathcal K$ (respectively, class-$\mathcal{K}_\infty$) functions, e.g., for $\alpha_1, \alpha_2\in \mathcal{GK}$ and $\alpha\in \mathcal{K}$, 
we have:
\begin{itemize}
    \item $\alpha_1\circ \alpha_2\in \mathcal{GK}$ and $\alpha_1+\alpha_2 \in \mathcal{GK}$;
    \item $\alpha\circ\alpha_1\in \mathcal{GK}$, $\alpha_1\circ\alpha\in \mathcal{GK}$ and $\alpha_1 + \alpha \in \mathcal{GK}$.
\end{itemize}


\noindent Consider the system: 
\begin{equation}\begin{split}\label{ex sys}
\dot y(t) = f(y(t)),
\end{split}\end{equation}
where $y\in \mathbb R^n$, $f: D \rightarrow \mathbb R^n$ is continuous on an open neighborhood $D\subseteq \mathbb R^n$ of the origin and $f(0)=0$. The origin is said to be an FTS equilibrium of \eqref{ex sys} if it is Lyapunov stable and \textit{finite-time convergent}, i.e., for all $y(0) \in \mathcal N \setminus\{0\}$, where $\mathcal N$ is some open neighborhood of the origin, $\lim_{t\to T} y(t)=0$, where $T = T(y(0))<\infty$  \cite{bhat2000finite}.
The authors also presented Lyapunov conditions for FTS of the origin of \eqref{ex sys}:

\begin{Theorem}[\cite{bhat2000finite}]\label{FTS Bhat}
Suppose there exists a continuous function $V$: $D \rightarrow \mathbb{R}$ such that the following holds: \\
(i) $V$ is positive definite \\
(ii) There exist real numbers $c>0$ and $\alpha \in (0, 1)$ , and an open neighborhood $\mathcal{V}\subseteq D$ of the origin such that 
\begin{equation}\begin{split} \label{FTS Lyap}
    \dot V(y) \leq - cV(y)^\alpha, \; y\in \mathcal{V}\setminus\{0\}.
\end{split}\end{equation}
Then the origin is an FTS equilibrium for \eqref{ex sys}.
\end{Theorem}

\subsection{Main result}
We consider the class of hybrid systems $\mathcal H = \{\mathcal F, \mathcal G , C, D\}$ described as
\begin{equation}\label{hybrid sys}
\begin{split}
    \dot x(t) & = f_{\sigma_f(t,x)}(x(t)), \quad x(t)\in C,\\
    x(t^+) & = g_{\sigma_g(t,x)}(x(t)), \quad x(t)\in D,
\end{split}\end{equation}
where $x\in \mathbb R^n$ is the state vector with $x(t_0) = x_0$, $f_{i}\in \mathcal F \triangleq\{f_k\}$ for $k\in \Sigma_f \triangleq\{1,2,\dots, N_f\}$ is the continuous flow (called thereafter, continuous-time mode, or simply, mode) allowed on the subset of the state space $C\subset\mathbb R^n$, and $g_{j} \in \mathcal G\triangleq\{g_l\}$ for $l\in \Sigma_g \triangleq\{1,2,\dots, N_g\}$ defines the discrete behavior (called thereafter discrete-jump dynamics), which is allowed on the subset $D\subset\mathbb R^n$. Define $x^+(t) \triangleq x(t^+)$. The switching signals $\sigma_f:\mathbb R_+\times\mathbb R^n\rightarrow\Sigma_f$ and $\sigma_g :\mathbb R_+\times\mathbb R^n\rightarrow\Sigma_g$ are assumed to be piecewise constant and right-continuous, in general dependent upon both state and time. We omit the argument $(t,x)$ from the functions $\sigma_f,\sigma_g$ for sake of brevity.


\begin{Remark}
Note that \eqref{hybrid sys} is a generalization of system (1.2) in \cite[Chapter 1]{goebel2012hybrid} that is given as:
\begin{equation}
\begin{split}
    \dot x(t) & = f(x(t)), \quad x(t)\in C,\\
    x(t^+) & = g(x(t)), \quad x(t)\in D,
\end{split}\end{equation}
that describes a hybrid system with one continuous flow $f$, and one discrete-jump dynamics $g$, i.e., $N_f = N_g = 1$. Furthermore, \eqref{hybrid sys} is a special case of system (1.1) in \cite[Chapter 1]{goebel2012hybrid}, given in terms of differential inclusion and difference inclusion as: 
\begin{equation}
\begin{split}
    \dot x(t) & \in F(x), \quad x(t)\in C,\\
    x(t^+) & \in G(x), \quad x(t)\in D,
\end{split}\end{equation}
where $F, G:\mathbb R^n\rightrightarrows\mathbb R^n$ are set-valued maps.
\end{Remark}

\begin{Assumption}\label{assum uniq equib}
The functions $f_i$ are continuous for all $i\in \Sigma_f$. The origin is the only equilibrium point for all the continuous flows and discrete jumps, {i.e., $f_i(x) = 0 \iff x = 0$ for all $i\in \Sigma_f$ and $g_j(x) = 0 \iff x = 0$ for all $j \in \Sigma_g$.}
\end{Assumption}

Per Assumption \ref{assum uniq equib}, we restrict our attention to the case where there is a unique equilibrium for the hybrid system \eqref{hybrid sys}. The case when there exists some $g_j\in \mathcal G$ and a set $\bar D\neq\{0\}$ such that $g_j(x) = 0$ for all $x\in \bar D\subset D$ can be treated by studying stability of set $\bar D$; see \cite{li2019finite,li2016results}. 

{A mode $F\in \Sigma_f$ is called an FTS subsystem or FTS mode if the origin of $\dot y = f_F(y)$ is FTS.} Denote by $T_{i_k} =  [t_{i_k},t_{i_k+1})$ the interval in which the flow $f_i$ is active for the $k-$th time for $i\in \Sigma_f$ and $k\in \mathbb N$, and $t^d_{j_m}$ the time when discrete jump $x^+ = g_j(x)$ takes place for the $m-$th time for $j\in \Sigma_g$ and $m\in \mathbb N$. Define $J_{i} = \{t^d_{j_m}\; | t^d_{j_m}\in T_{i_k}, j\in \Sigma_g, m\in \mathbb N\}$ as the set of all time instances when a discrete jump takes place when the continuous flow $f_i$ is active. Without loss of generality, we assume that the switching signals $\sigma_f$ and $\sigma_g$ are minimal, i.e., for any $i\in \Sigma_f$, $t_{i_{k+1}}\neq t_{i_k+1}$ for all $k\in \mathbb R_+$, and that there no two discrete-jumps at the same time instant. Inspired by \cite[Definition 1]{sun2012integral} and \cite[Definition 2.6]{goebel2012hybrid}, let us define the concepts of the solution of the hybrid system \eqref{hybrid sys} as follows (interested reader is referred to \cite[Chapter 2]{goebel2012hybrid} for detailed presentation on solution concept of hybrid systems). 

\begin{Definition}\label{Def sol hyb sys}
Let $\phi:\mathbb R_+\times\mathbb Z_+\rightarrow\mathbb R^n$ satisfy the following:
\begin{itemize}
    \item for all $j\in \mathbb Z_+$, 
    \begin{itemize}
        \item $\phi(t, j)$ is continuously differentiable for all $t\in \bigcup_k T_{l_k}\setminus(\bigcup_i J_i)$ for all $l\in \Sigma_f$;
        \item $\phi(\cdot, j)$ is absolutely continuous on $\mathbb R_+\setminus (\bigcup_i J_i)$ and right-continuous on $\mathbb R_+$;
        \item $\phi(t, j)$ satisfies $\dot \phi(t,j) = f_{\sigma_f}(\phi(t,j))$, for all $\phi(t,j)\in C$, and for all $t\in \mathbb R_+\setminus (\bigcup_i J_i)$;
    \end{itemize}
    \item for all $t\in \bigcup_i J_i$, $\phi(t, j)$ satisfies $\phi(t, j+1) = g_{\sigma_g}(\phi(t, j))$, for all $\phi(t, j)\in D$, and for all $j\in \mathbb Z_+$;
\end{itemize}
Then, the projection of $\phi(\cdot,\cdot)$ on its first argument, i.e., on continuous-time axis, is a solution of \eqref{hybrid sys}. In other words, a function $x:\mathbb R_+\rightarrow\mathbb R^n$, defined as $x(t) = \phi(t, j)$ for all $j\in \mathbb N$, $t\in \mathbb R_+\setminus(\bigcup_i J_i)$, and $x(t^+) =\phi(t,j+1)$, for all $j\in \mathbb N$ and for all $t\in \bigcup_i J_i$, is called as a solution of \eqref{hybrid sys}. 
\end{Definition}

\begin{Remark}
In \cite{goebel2012hybrid}, the solutions of hybrid systems are described using a hybrid arc $\phi:\mathbb R_+\times\mathbb Z_+\rightarrow\mathbb R^n$, which is parameterized by continuous-time $t\in \mathbb R_+$ and discrete-jump time $j\in \mathbb Z_+$. Since we are only concerned about the stability in the continuous-time $t$, we define the solution of \eqref{hybrid sys} in Definition \ref{Def sol hyb sys} as the projection of the function $\phi$ on the continuous-time axis. 
\end{Remark}

Let $\textnormal{dom}\;\phi\subset\mathbb R_+\times\mathbb Z_+$ denote the domain of definition of function $\phi$. Based on the structure of $\textnormal{dom}\;\phi$, the solutions of \eqref{hybrid sys} can be characterized in various ways, as discussed below.

\begin{Definition}\label{rem sol hyb sys}
The solution of \eqref{hybrid sys} is
\begin{itemize}
    \item \textbf{non-trivial}, if $\textnormal{dom}\;\phi$ contains at least one more point different from $(0,0)$;
    \item \textbf{complete}, if $\textnormal{dom}\;\phi$ is unbounded;
    \item \textbf{Zeno}, if it is complete but the projection of $\textnormal{dom}\;\phi$ on $\mathbb R_+$ is bounded;
    \item \textbf{maximal}, if $\textnormal{dom}\;\phi$ cannot be extended;
    \item \textbf{continuous}, if non-trivial, and $\textnormal{dom}\;\phi\subset \mathbb R_+\times\{0\}$;
    \item \textbf{discrete}, if non-trivial, and $\textnormal{dom}\;\phi\subset \{0\}\times\mathbb Z_+$;
    \item \textbf{eventually continuous}, if non-trivial and $J = \sup_j\textnormal{dom}\;\phi<\infty$;
    \item \textbf{eventually discrete}, if non-trivial and $T = \sup_t\textnormal{dom}\;\phi<\infty$;
    \item \textbf{compact}, if $\textnormal{dom}\;\phi$ is compact.
\end{itemize} 
\end{Definition}

\noindent Before presenting the main result, we make the following assumption on the solution of \eqref{hybrid sys}. Similar assumptions have been used in literature (e.g., \cite{zhao2008stability,li2019finite,sanfelice2007invariance}) in order to analyze stability properties of the origin of hybrid systems. 

\begin{Assumption}\label{assum exist sol conc}
The solution of \eqref{hybrid sys} exists, is non-Zeno, non-trivial and complete. 
\end{Assumption}

\noindent Assumption \ref{assum exist sol conc}, in light of Definitions \ref{Def sol hyb sys} and \ref{rem sol hyb sys}, implies that the solution $x(\cdot)$ of \eqref{hybrid sys} is defined for all times, is continuously differentiable while evolving along any of the continuous flows $f_i$, is absolutely continuous between any two discrete jumps, and is right-continuous at all times.

For each interval $T_{i_k}$, define the largest connected sub-interval $\bar T_{i_k}\subset T_{i_k}$, such that there is no discrete jump in $\bar T_{i_k}$, i.e., $\textrm{int}(\bar T_{i_k})\bigcap J_i = \emptyset$. For example, if $T_{i_1} = [0,1)$ and $J_{i} = \{0.2, 0.4, 0.75\}$, then $\bar T_{i_1} = [0.4, 0.75)$. 

\begin{Assumption}\label{assum bar T tm}
For mode $F\in\Sigma_f$, the length $|\bar T_{F_k}|$ of the time interval $\bar T_{F_k}$ satisfies $|\bar T_{F_k}| \geq t_d>0$ for all $k\in \mathbb N$. 
\end{Assumption}

\noindent Assumption \ref{assum bar T tm} implies that for the FTS flow $f_F$, in each interval $T_{F_k}$ when the system \eqref{hybrid sys} evolves along the flow $f_F$, there exists a sub-interval $\bar T_{F_k}$ of non-zero length $t_d$ such that there is no discrete jump in the system state during $\bar T_{F_k}$.

We first define the notion of FTS for hybrid systems. The standard notion of stability under arbitrary switching, as employed in \cite{liberzon2003switching,lin2009stability,branicky1998multiple,zhao2008stability,fu2015global}, is restrictive in the following sense. The conditions therein require every single mode of the system \eqref{hybrid sys} to be Lyapunov Stable (LS or simply, stable), Asymptotically Stable (AS), or FTS for the origin of the system \eqref{hybrid sys} to be LS, AS, or FTS, respectively. We overcome this limitation by defining the corresponding notions of stability and \textit{uniform} stability for hybrid system (inspired in part, from \cite[Theorem 1]{peleties1991asymptotic}) as following. Let $\Pi\subset \textnormal{PWC}(\mathbb R_+\times \mathbb R^n,\Sigma_f\times \Sigma_g)$ denote the set of all possible pairs of switching signals, where $\textnormal{PWC}$ is the set of all piecewise constant functions mapping from $\mathbb R_+\times\mathbb R^n$ to $\Sigma_f\times \Sigma_g$.

\begin{Definition}\label{Switch FTS Def}
The origin of the hybrid system \eqref{hybrid sys} is called \textbf{LS, AS or FTS} if there exists an open neighborhood $D\subset \mathbb R^n$ such that for all $y\triangleq x(0)\in D$, there exists a subset of switching signals $\Pi_y\subset \Pi$ such that the origin of the system \eqref{hybrid sys} is LS, AS or FTS, respectively, with respect to $(\sigma_f,\sigma_g)$, for all $(\sigma_f,\sigma_g)\in \Pi_y$. If $\Pi_y = \Pi$ for all $y\in D$, then the origin is said to be \textbf{uniformly LS, AS or FTS}. 
\end{Definition} 

Per Definition \ref{Switch FTS Def}, FTS of the origin of \eqref{hybrid sys} is realized under a (given set of) switching signal(s), while uniform FTS of the origin is realized under any arbitrary switching signal. Note that the aforementioned papers use the latter notion of uniform stability in their analysis.


Let $\bar T_{F_k} = [\bar t_{F_k},\bar t_{F_k+1})$ with $ \bar t_{F_k+1}-\bar t_{F_k}\geq t_d$, and $\{\bar V_{F_1}, \bar V_{F_2}, \dots, \bar V_{F_p}\}$ and $\{\bar V_{F_1+1}, \bar V_{F_2+1}, \dots, \bar V_{F_p+1}\}$ be the sequence of the values of the Lyapunov function $V_F$ at the beginning and at end of the intervals $\bar T_{i_k}$, respectively. Let $\{i^0, i^1, \dots,i^l,\ldots\}\in\Sigma_f$ be the sequence of modes that are active during the intervals $[t_{0}, t_{1}), [t_{1},t_{2}), \dots,[t_l,t_{l+1}),\ldots\,$, respectively.
We state the following result before we proceed to the main theorem. 
\begin{Lemma}\label{V diff lemma}
Let $a_i, b_i\geq 0$ are such that $a_i\geq b_i$ for all $i\in \{1, 2, \dots, K\}$ for some $K\in \mathbb N$. Then, for any $0<r<1$, we have 
\begin{align}
    \sum_{i = 1}^K(a_i^r-b_i^r)\leq \sum_{i = 1}^K(a_i-b_i)^r.
\end{align}
\end{Lemma}
\noindent The proof is given in  Appendix \ref{app lemma 1 pf}.
We now present our main result on FTS of hybrid systems.

\begin{Theorem}\label{FT Th 2}
The origin of \eqref{hybrid sys} is LS if there exist Lyapunov functions $V_i$ for each $i\in \Sigma_f$ such that the following hold:
\begin{itemize}
 \item[(i)] There exists $\alpha_1 \in \mathcal{GK}$, such that 
\begin{equation}
\begin{split}
   \sum\limits_{k = 0}^{p}\Big(
   V_{i^{k+1}}(x(t_{{k+1}})) -&V_{i^k}(x(t_{{k+1}}))\Big)\leq \alpha_1(\|x_0\|), \label{Hyb cond 1}
\end{split}\end{equation}
holds for all $p\in \mathbb Z_+$;
\item[(ii)] There exists $\alpha_2\in \mathcal{GK}$ such that
\begin{equation}\begin{split}
   \sum\limits_{k = 0}^{p}\Big(
   V_{i^k}(x(t_{{k+1}})) -&V_{i^k}(x(t_{k}))\Big)\leq \alpha_2(\|x_0\|), \label{Hyb cond 2}
\end{split}\end{equation}
holds for all $p\geq 0$;
\item[(iii)] There exists $\alpha_3\in \mathcal{GK}$ such that for all $i\in \Sigma_f$, 
\begin{equation}\begin{split}
    \sum\limits_{t\in J_{i}}\Big(V_i(x^+(t))-V_i(x(t))\Big)& \leq \alpha_3(\|x_0\|); \label{Hyb cond 3}.
\end{split}\end{equation}
If, in addition, there exist switching signals $(\sigma_f,\sigma_g)$ and,
\item[(iv)] There exists $F\in \Sigma_f$ such that the origin of $\dot y = f_F(y)$ is FTS, and there exist a positive definite Lyapunov function $V_F$ and constants $c>0$, $0<\beta<1$ such that 
\begin{equation}\label{v dot cond}
\dot V_F \leq -cV_F^\beta,    
\end{equation}
for all $t\in \bigcup [t_{F_k},t_{F_k+1})\setminus J_F$;
\item[(v)] The accumulated duration $\bar T_F \triangleq \sum_k{\bar T}_{F_k}$ corresponding to the period of time during which the mode $F$ is active without any discrete jumps, satisfies
\begin{align*}
\bar T_F = \gamma (\|x_0\|) \triangleq  \frac{\alpha(\|x\|)^{1-\beta}}{c(1-\beta)} + \frac{\bar \alpha(\|x\|)^{1-\beta}}{c(1-\beta)},
\end{align*}
where $\alpha = \alpha_0+\alpha_1+\alpha_2+N_f\alpha_3$, $\bar \alpha = \alpha_1+\alpha_2+N_f\alpha_3$ and $\alpha_0 \in \mathcal{GK}$,
\end{itemize}
then, the origin of \eqref{hybrid sys} is FTS with respect to the switching signal $(\sigma_f,\sigma_g)$. Moreover, if all the conditions hold globally, the functions $V_i$ are radially unbounded for all $i\in \Sigma_f$, and $\alpha_l\in \mathcal{GK}_\infty$ for $l\in \{1,2,3\}$, then the origin of \eqref{hybrid sys} is globally FTS. 
\end{Theorem}

\begin{proof}
First we prove the stability of the origin under conditions (i)-(iii). Let $x_0\in D$, where $D$ is some open neighborhood of the origin. For all $p\in \mathbb Z_+$, we have that
\begin{align*}
    V_{i^p}(x(t_p)) = & V_{i^0}(x(t_{0})) + \sum\limits_{k = 1}^{p}
   \Big(V_{i^k}(x(t_{k})) -V_{i^{k-1}}(x(t_{k}))\Big) \\
   & + \sum\limits_{k = 0}^{p-1}
   \Big(V_{i^k}(x(t_{{k+1}})) -V_{i^k}(x(t_k))\Big)\\
   & + \sum\limits_{k = 0}^{p}\sum_{t\in J_{k}\bigcap [t_k, t_{k+1})}
   \Big(V_{i^k}(x(t^+)) -V_{i^k}(x(t))\Big)\\
    \overset{\eqref{Hyb cond 1},\eqref{Hyb cond 2},\eqref{Hyb cond 3}}{\leq}& \alpha(\|x_0\|)
\end{align*}
where $\alpha = \alpha_0 + \alpha_1+\alpha_2 + N_f\alpha_3$ with $\alpha_0(r) = \max\limits_{i\in \Sigma_f, \; \|x\|\leq r}V_i(x)$.
\begin{figure}[ht!]
	\centering
	\includegraphics[width=0.8\columnwidth,clip]{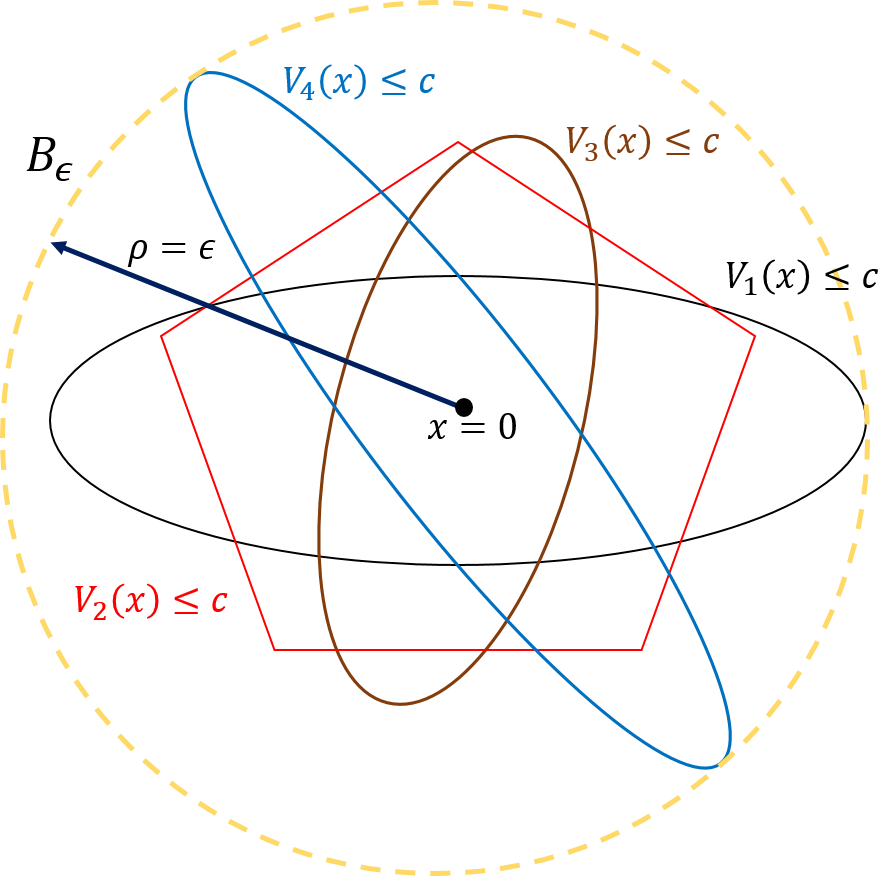}
	\caption{The ball $B_\rho$, shown in dotted yellow, encloses $c$ sublevel sets of the Lyapunov functions $V_i$, whose boundaries are shown in solid lines.}
	\label{fig:V i enclos}
\end{figure}
Thus, we have:
\begin{align}\label{v_ip ineq stability}
    V_{i^p}(x(t_p))\leq \alpha(\|x_0\|),
\end{align}
for all $p\in \mathbb Z_+$. 
Let $d_i(c) = \{x \; |\; V_i(x) \leq c\}$ denote the $c$ sub-level set of the Lyapunov function $V_i$, $i\in \Sigma_f$, and $B_\rho = \{x\; |\; \|x\| \leq \rho\}$ denote a ball centered at the origin with radius $\rho\in \mathbb R_+$. Define $r(c) =  \inf\{\rho \geq 0\; |\; d_i(c)\subset B_\rho\}$ as the radius of the smallest ball centered at the origin that encloses the $c$ sub-level sets $d_i(c)$, for all $i\in \Sigma_f$ (see Figure \ref{fig:V i enclos}). Since the functions $V_i$ are positive definite, the sub-level sets $d_i(c)$ are bounded for small $c>0$, and hence, the function $r$ is invertible. The inverse function $c_\epsilon = r^{-1}(\epsilon)$ maps the radius $\epsilon>0$ to the value $c_\epsilon$ such that the sub-level sets $d_i(c_\epsilon)$ are contained in $B_\epsilon$ for all $i\in \Sigma_f$. For any given $\epsilon>0$, choose $\delta = \alpha^{-1}(r^{-1}(\epsilon))>0$ so that \eqref{v_ip ineq stability} implies that for $\|x_0\|\leq \delta$, we have
\begin{align*}
    V_{i^p}(x(t_p)) &\leq \alpha(\|x_0\|) \leq \alpha(\alpha^{-1}(r^{-1}(\epsilon))) = r^{-1}(\epsilon)\\
    \implies 
    \|x(t_p)\| &\leq \epsilon,
\end{align*}
for all $p\in \mathbb Z_+$, i.e., the origin is LS. 

Next, we prove FTS of the origin when conditions (iv)-(v) also hold. 
From \eqref{v_ip ineq stability}, we have that 
\begin{align}\label{v f a}
 V_F(x(t_{F_i})) \leq \alpha(\|x_0\|), 
\end{align}
for all $i\in \mathbb N$. By definition, we have that there is no discrete jump during $\bar T_{F_k}$, for all $k\in \mathbb N$. Let $M\in \mathbb N$ denote the total number of times the mode $F$ is activated. From condition (iv), we have $\dot V_F \leq -cV_F^\beta$ for all $t\in \bigcup\bar T_{F_k}$. Using this, we obtain that for any $M$, we have
\begin{equation*}\begin{split}
|\bar T_{F_k}|&  \leq \frac{\bar V_{F_k}^{1-\beta}}{c(1-\beta)}-\frac{\bar V_{F_k+1}^{1-\beta}}{c(1-\beta)} \\
\implies \sum\limits_{k = 1}^M |\bar T_{F_k}| & \leq \sum\limits_{k = 1}^M\Big(\frac{\bar V_{F_k}^{1-\beta}}{c(1-\beta)}-\frac{\bar V_{F_k+1}^{1-\beta}}{c(1-\beta)}\Big)\\
& = \frac{\bar V_{F_1}^{1-\beta}}{c(1-\beta)} + \sum\limits_{i = 1}^
{M-1}\frac{\bar V_{F_{i+1}}^{1-\beta}-\bar V_{F_i+1}^{1-\beta}}{c(1-\beta)} -\frac{\bar V_{F_M+1}^{1-\beta}}{c(1-\beta)}. 
\end{split}\end{equation*}
Using \eqref{v f a}, we obtain that 
\begin{align}
    \frac{V_{F_1}^{1-\beta}}{c(1-\beta)} \leq \frac{\alpha(\|x_0\|)^{1-\beta}}{c(1-\beta)}.
\end{align} 
Define $\gamma_1(\|x_0\|) \triangleq \frac{\alpha(\|x_0\|)^{1-\beta}}{c(1-\beta)}$ so that $\gamma_1\in \mathcal{GK}$. Now, let $I_1 = \{i_1, i_2, \dots, i_k\}, 0\leq i_l\leq M,$ be the set of indices such that $\bar V_{F_{i+1}}\geq \bar V_{F_i+1}$ for $i\in I_1$. We know that for $a, b\geq 0$, $a\geq b \implies a^r\geq b^r$ for $r> 0$. Hence, we have that
\begin{align}
\sum\limits_{i = 1}^{M-1}\frac{\bar V_{F_{i+1}}^{1-\beta}-\bar V_{F_i+1}^{1-\beta}}{c(1-\beta)} \leq \sum\limits_{i \in I_1}\frac{\bar V_{F_{i+1}}^{1-\beta}-\bar V_{F_i+1}^{1-\beta}}{c(1-\beta)}
\end{align}
Using Lemma \ref{V diff lemma}, we obtain that 
\begin{align}
\sum\limits_{i \in I_1}\frac{\bar V_{F_{i+1}}^{1-\beta}-\bar V_{F_i+1}^{1-\beta}}{c(1-\beta)}\leq \sum\limits_{i \in I_1}\frac{(\bar V_{F_{i+1}}-\bar V_{F_i+1})^{1-\beta}}{c(1-\beta)}.
\end{align}
From the analysis in the first part of the proof, we know that 
\begin{align*}
    \bar V_{F_{i+1}}-\bar V_{F_i+1} & = \sum\limits_{k = l_1}^{l_2-1}\Big(
   V_{i^k}(x(t_{{k}})) -V_{i^{k-1}}(x(t_{{k}}))\Big)\\
   & +\sum\limits_{k = l_1}^{l_2-1}\Big(
   V_{i^k}(x(t_{{k+1}}))-V_{i^k}(x(t_{k}))\Big)\\
   & + \sum\limits_{k = l_1}^{l_2-1}\sum_{t\in J_{k}\bigcap [t_k, t_{k+1})}\Big(V_i(x^+(t))-V_i(x(t))\Big),
\end{align*}
where $l_1, l_2$ are such that $t_{l_1}$ denotes the time when mode $F$ becomes deactivated for the $i$-th time and $t_{l_2}$ denotes the time when the mode $F$ is activated for $(i+1)$-th time. Define $\bar \alpha = \alpha_1+\alpha_2 + N_f\alpha_3$ so that we have 
\begin{align*}
    \sum_{i\in I_1}\bar V_{F_{i+1}}-\bar V_{F_i+1} \leq \bar \alpha(\|x_0\|). 
\end{align*}
Hence, we have that
\begin{align}\label{v f need}
\sum\limits_{i = 1}^{M-1}\frac{\bar V_{F_{i+1}}^{1-\beta}-\bar V_{F_i+1}^{1-\beta}}{c(1-\beta)} \leq \frac{\sum\limits_{i \in I_1}(\bar V_{F_{i+1}}-\bar V_{F_i+1})^{1-\beta}}{c(1-\beta)}\leq \frac{\bar \alpha(\|x_0\|)^{1-\beta}}{c(1-\beta)}.
\end{align}
Define $\gamma(\|x_0\|) \triangleq \gamma_1(\|x_0\|) + \frac{\bar \alpha(\|x_0\|)^{1-\beta}}{c(1-\beta)}$ so that we obtain:
\begin{equation*}\begin{split}
   \bar T_F + \frac{\bar V_{F_M+1}^{1-\beta}}{c(1-\beta)} &\leq \frac{\bar V_{F_1}^{1-\beta}}{c(1-\beta)} + \sum\limits_{i = 1}^
{M-1}\frac{\bar V_{F_{i+1}}^{1-\beta}-\bar V_{F_i+1}^{1-\beta}}{c(1-\beta)}\leq  \gamma(\|x_0\|).
\end{split}\end{equation*}
Clearly, $\gamma \in \mathcal{GK}$. Now, with $\bar T_F  = \gamma(\|x_0\|)$, we obtain 
\begin{equation*}\begin{split}
     \bar T_F+\frac{\bar V_{F_M+1}^{1-\beta}}{c(1-\beta)}\leq \gamma(\|x_0\|) =  \bar T_F,
\end{split}\end{equation*}
which implies that $\frac{\bar V_{F_M+1}^{1-\beta}}{c(1-\beta)} \leq 0$. However, $\bar V_F\geq 0$, which further implies that $\bar V_{F_M+1} = 0$. Hence, if mode $F$ is active for the accumulated time $\bar T_F$ without any discrete jump in the system state, the value of the function $V_F$ converges to $0$ as $t\rightarrow \bar t_{F_M+1}$. 

From Assumption \ref{assum bar T tm}, $|\bar T_{F_k}|\geq t_d$ for all $k\in \mathbb N$, and hence $Mt_d \leq \sum\limits_{i = 1}^M|\bar T_{F_k}| =  \gamma(\|x_0\|)$, which implies that $M\leq \frac{\gamma(\|x_0\|)}{t_d} <\infty$ (i.e., the number of times the mode $F$ is activated is finite). Next we show that the time of convergence is also finite, i.e., $\bar t_{F_M+1} <\infty$. From the above analysis, we have that if $\sum\limits_{i = 1}^M |\bar T_{F_i}| =  \bar T_F$, then there exists an interval $[\bar t_{F_M},\bar t_{F_M+1})$ such that $\bar V_{F_M+1} = 0$. Since $\bar V_{F_i} \leq \bar\alpha(\|x_0\|)<\infty$, we obtain that
\begin{align}
\bar t_{F_M+1}-\bar t_{F_M} \leq \frac{\bar V_{F_M}^{1-\beta}-\bar V_{F_M+1}^{1-\beta}}{c(1-\beta)}\leq \frac{\bar V_{F_M}^{1-\beta}}{c(1-\beta)} <\infty , 
\end{align}
for all $x_0\in D$. Now, there are two cases possible. If $\bar t_{F_M}<\infty$, then, we obtain that $\bar t_{F_M+1} \leq \bar t_{F_M} + \frac{\bar\alpha(\|x_0\|)^{1-\beta}}{c(1-\beta)}<\infty$ for all $x_0\in D$. If $\bar t_{F_M} = \infty$, we obtain that the time of activation for mode $F$ is $\bar T_F = \sum\limits_{i = 1}^{M-1}|\bar T_i| <\gamma(\|x_0\|)$, which however contradicts condition (v). Thus, for condition (v) to hold, it is required that $\bar t_{F_M}<\infty$ and therefore $\bar t_{F_M+1}<\infty$. Hence, the trajectories of \eqref{hybrid sys} reach the origin within a finite number of active intervals of the continuous flow $f_F$. This proves that the origin is FTS.

Finally, if all the conditions (i)-(v) hold globally and the functions $V_i$ are radially unbounded, we have that $\alpha_0$ is also radially unbounded. Since $\alpha_1, \alpha_2, \alpha_3,\alpha_4\in \mathcal{GK}_\infty$, we have $\bar\alpha(\|x_0\|) <\infty$ for all $\|x_0\|<\infty$, and hence, we obtain $t_{F_M+1} \leq t_{F_M} + \frac{\bar\alpha(\|x_0\|)^{1-\beta}}{c(1-\beta)}<\infty$ for all $x_0$, which implies global FTS of the origin.  
\end{proof}

\begin{Remark}
Theorem \ref{FT Th 2} essentially says that a set of sufficient conditions for FTS of the origin of \eqref{hybrid sys} are as follows: a) the origin is uniformly LS; b) there exists a FTS mode $F\in \Sigma_f$, and a function $V_F$ that satisfies \eqref{v dot cond}; and that c) the FTS mode $F$ is active for a sufficient amount of cumulative time, which depends upon the initial conditions. This is formally stated in the following corollary. 
\end{Remark}

\begin{Cor}\label{cor LS FTS mode}
Suppose that the origin of \eqref{hybrid sys} is uniformly LS, and that there exists an FTS mode $F\in \Sigma_f$ and a corresponding positive definite function $V_F$ satisfying \eqref{v dot cond}. Then, the origin is FTS if there exists a switching signal $(\sigma_f,\sigma_g)$ such that the FTS mode $F$ is active for a cumulative time of $\bar \gamma(\|x_0\|)$, where $\bar\gamma\in \mathcal{GK}$. 
\end{Cor}

The proof is given in Appendix \ref{app:proof cor LS FTS mode}. In light of this observation, Theorem \ref{FT Th 2} can be further interpreted as: If uniform stability of the origin can be established, then the presence of a switching signal and an FTS mode such that the latter is active for a sufficient amount of time is sufficient to guarantee FTS of the origin for the overall system. Note that uniform stability, and not just stability, of the origin is required in the above result as there might exist cases where the origin is stable for a particular pair of switching signals $(\sigma_f,\sigma_g)$ such that FTS mode $F$ is not active at all, and switching to the FTS mode leads to instability (see \cite{branicky1998multiple} for an example of the case where introducing an AS mode results into instability of the origin that is otherwise AS). Uniform stability ensures that the origin is stable under arbitrary switching signals, ruling out such a possibility.

To assess uniform Lyapunov stability of the origin, one can use either the conditions in terms of multiple generalized Lyapunov functions in \cite{zhao2008stability}, or conditions in terms of a common Lyapunov function \cite{liu2016lyapunov,goebel2016notions}; then, the conditions (iv)-(v) of Theorem \ref{FT Th 2} can be checked independently to establish FTS of the origin. 

\begin{Remark}
In practice, the conditions (i)-(iii) or those presented in \cite{zhao2008stability} can be difficult to verify for a general class of hybrid system involving non-linear subsystems; the study of finding Lyapunov functions to assess stability for hybrid systems is an open field of research, and is out of scope of this work. In Section \ref{sec switch signal}, we present a method of designing switching signal $\sigma_f$ and functions $V_i$ for a class of switched linear systems. 
\end{Remark}




\subsection{Discussion on the main result}

\begin{figure}[ht!]
	\centering
	\includegraphics[width=0.9\columnwidth,clip]{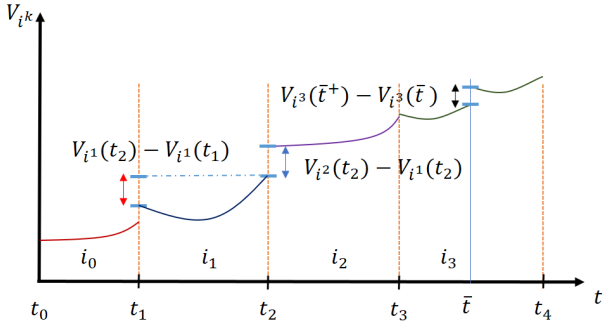}
	\caption{Conditions (i), (ii) and (iii) of Theorem \ref{FT Th 2} regarding the allowable changes in the values of the  Lyapunov functions. The increments shown by blue, red and black double-arrows pertain to condition (i), (ii) and (iii), respectively. }
	\label{fig:V i incr}
\end{figure}

\textbf{Intuitive explanation of the conditions of Theorem \ref{FT Th 2}}: Condition (i) means that at switching instants of the dynamics of continuous flows (i.e., at switches in signal $\sigma_f$), the cumulative value of the differences between the consecutive Lyapunov functions is bounded by a class-$\mathcal{GK}$ function. Condition (ii) means that the cumulative increment in the values of the individual Lyapunov functions when the respective modes are active, is bounded by a class-$\mathcal{GK}$ function (see Figure \ref{fig:V i incr}).\footnote{Note that some authors use the time derivative condition, i.e., $\dot V_i\leq \lambda V_i$ with $\lambda>0$, in place of condition (ii), to allow growth of $V_i$, hence, requiring the function to be continuously differentiable (e.g. \textnormal{\cite{wang2018conditions}}). Our condition allows the use of non-differentiable  Lyapunov functions.} Condition (iii) means that the cumulative increment in the value of the  Lyapunov function $V_i$ is bounded by a class-$\mathcal{GK}$ function at the discrete jumps. Condition (iv) means that there exists an FTS mode $F\in \Sigma$. Finally, condition (v) means that the FTS mode $F$ is active for a sufficiently long cumulative time $\gamma(\|x_0\|)$ without any discrete jump occurring in that cumulative period. Depending upon the application at hand, and available authority on the design of the switching signal, the FTS mode can be made active for $\gamma(\|x_0\|)$ duration in one activation period, or in multiple activation periods. 


\textbf{Comparison with earlier results}: In contrast to \cite[Proposition 3.8]{zhao2008stability}, where the authors provided necessary and sufficient conditions for stability of switched system under \eqref{v_ip ineq stability} and non-increasing condition on the Lyapunov functions during activation period, we proved stability of the origin with just \eqref{v_ip ineq stability}. Compared to \cite{li2019finite,li2016results}, our results are less conservative in the sense that the  Lyapunov functions are allowed to increase during the continuous flows (per \eqref{Hyb cond 2}), as well as at the discrete jumps (per \eqref{Hyb cond 3}). In other words, we allow unstable modes as well as unstable discrete-jumps (i.e., discrete jump $x^+ = g(x)$ such that $V(g(x))>V(x)$ for some positive definite function $V$) to be present in the hybrid system while still guaranteeing FTS of the origin. Also, during the continuous flows, the  Lyapunov functions are allowed to grow when switching from one continuous flow to another (per \eqref{Hyb cond 1}), whereas the aforementioned work imposes that the common Lyapunov function is always non-increasing. In contrast to some of the previous work e.g., \cite{li2019finite,zhang2018finite,wang2018conditions,cai2008smooth}, except for $V_F$ we do not require the Lyapunov functions to be differentiable.

\textbf{The usefulness of Lemma 1}: In general, the inequality \eqref{v f need} can be difficult to obtain directly. Consider the case when we only know that the mode $F$ is homogeneous, with negative degree of homogeneity. From \cite[Theorem 7.2]{bhat2005geometric}, we know that condition (iii) of Theorem \ref{FT Th 2} holds for some $\beta$, but its exact value might not be known. In this case, it is not possible to bound the left-hand side (LHS) of \eqref{v f need}. Lemma \ref{V diff lemma} allows one to bound this LHS with a class-$\mathcal{GK}$ function without explicitly knowing the value of $\beta$. 

\textbf{Remarks on Assumption \ref{assum bar T tm}}: We discuss the relation of Assumption \ref{assum bar T tm} with the average dwell-time (ADT) for discrete jumps method that is often used in the literature \cite{teel2013lyapunov}. Discrete jumps under ADT means that in any given interval $[t_1, t_2]$, the number of discrete jumps of the states $N([t_1,t_2])$ satisfies $N([t_1,t_2])\leq N_0 + \delta(t_2-t_1)$, where $\delta\geq 0$ and $N_0\geq 1$. (see \cite{teel2013lyapunov,cai2008smooth}). We note that if the conditions of the ADT method hold for the mode $F$, instead of the minimum dwell-time condition of Assumption \ref{assum bar T tm}, then the FTS result still holds since the parameter $M$ used in the proof of Theorem \textnormal{\ref{FT Th 2}} can be defined as $M = N_0 + \delta T_F$, where $T_F$ is the total time of activation of mode $F$.

Note that if Assumption \ref{assum bar T tm} does not hold, one can construct a counter-example where the system would need to execute a Zeno behavior in order to achieve FTS. Consider the case when $|\bar T_{F_k}|= \frac{\bar T_F}{2^k}$. It is clear that for $\sum\limits_{k = 1}^M|\bar T_{F_k}| \rightarrow \bar T_F$, we need $M\rightarrow\infty$. Hence, this would require infinite number of switches from some mode $i\neq F$ to the mode $F$ in a finite amount of time $\bar T_F$, i.e., the system would need to execute a Zeno behavior to achieve FTS. With Assumption \ref{assum bar T tm}, we obtain that $M \leq \frac{\bar T_F}{t_d} <\infty$, which rules out this possibility. 




\section{FTS of Switched Systems}\label{sec Switch System}

\subsection{FTS result}

In this subsection,  we illustrate how the case of switched systems, i.e., of systems without discrete jumps in their states, is a special case of the results derived above. In summary, in the case of a switched system, Theorem \ref{FT Th 2} guarantees FTS of the origin under conditions (i), (ii), (iv) and (v); condition (iii) is obsolete since $D = \emptyset$. As a side note, if in addition to $D=\emptyset$, one has that $N_f = 1$, i.e., if the system \eqref{hybrid sys} reduces to a continuous-time dynamical system, Theorem \ref{FT Th 2} reduces to Theorem \ref{FTS Bhat}. Thus, the seminal result on FTS of continuous-time systems is a special case of Theorem \ref{FT Th 2}. 

Consider the system
\begin{equation}\begin{split}\label{switch sys}
    \dot x(t) = f_{\sigma(t,x)}(x(t)), \quad x(t_0) = x_0,
\end{split}\end{equation}
where $x\in \mathbb R^n$ is the system state, $\sigma: \mathbb R_+\times\mathbb R^n\rightarrow \Sigma$ is a piecewise constant, right-continuous switching signal that can depend both upon state and time, $ \Sigma \triangleq \{1, 2, \dots, N\}$ with $N<\infty$, and $f_{\sigma(\cdot,\cdot)}:\mathbb R^n\rightarrow\mathbb R^n$ is the system vector field describing the active subsystem (called thereafter mode) under $\sigma(\cdot, \cdot)$. Note that \eqref{switch sys} is a special case of \eqref{hybrid sys} with $D = \emptyset$. Hence, the solution of \eqref{switch sys} is given by Definition \ref{Def sol hyb sys}, and does not exhibit discrete jumps, i.e., the solution is continuous. Similarly, FTS of the origin of \eqref{switch sys} is defined per Definition \ref{Switch FTS Def} with $\Pi\subset \textrm{PWC} (\mathbb R_+\times \mathbb R^n,\Sigma)$. We make the following assumption for \eqref{switch sys}. 
\begin{Assumption}\label{assum: dwell time}
The solution of \eqref{switch sys} exists and satisfies Assumption \ref{assum exist sol conc}. In addition, there is a non-zero dwell-time for the FTS mode $F\in \Sigma$, i.e., $|T_{F_k}| = t_{F_k+1}-t_{F_k}\geq t_d$ for all $k\in \mathbb N$, where $t_d>0$ is a positive constant.
\end{Assumption}
Note that in the absence of discrete jumps, Assumption \ref{assum bar T tm} results into $|\bar T_{F_k}| = |T_{F_k}| = t_{F_k+1}-t_{F_k}\geq t_d$ for $k\in \mathbb N$. Hence, Assumption \ref{assum: dwell time} is a special case of Assumption \ref{assum bar T tm}. 
We present the conditions for FTS of the origin of \eqref{switch sys} in terms of multiple Lyapunov functions. Let $\{i^0, i^1, \dots,i^p,\ldots\}$ be the sequence of modes that are active during the intervals $[t_{0}, t_{1}), [t_{1},t_{2}), \dots,[t_p,t_{p+1}),\ldots\,$, respectively, for $i^p\in\Sigma, p\in \mathbb Z_+$. 

\begin{Cor}\label{FT Th 3}
The origin of \eqref{switch sys} is LS if there exist a switching signal $\sigma$ and Lyapunov functions $V_i$ for each $i\in \Sigma$, and the following hold:
\begin{itemize}
   \item[(i)] There exists $\alpha_1 \in \mathcal{GK}$, such that 
\begin{equation}
\begin{split}
   \sum\limits_{k = 0}^{p}
   \Big(V_{i^{k+1}}(x(t_{{k+1}})) -&V_{i^k}(x(t_{{k+1}}))\Big)\leq \alpha_1(\|x_0\|), \label{cond 1}
\end{split}\end{equation}
holds for all $p\in \mathbb Z_+$;
\item[(ii)] There exists $\alpha_2 \in \mathcal{GK}$, such that 
\begin{equation}\begin{split}
   \sum\limits_{k = 0}^{p}
   \Big(V_{i^k}(x(t_{{k+1}})) -&V_{i^k}(x(t_{k}))\Big)\leq \alpha_2(\|x_0\|), \label{cond 2}
\end{split}\end{equation}
holds for all $p\in \mathbb Z_+$.
If in addition,
\item[(iii)] There exist a mode $F\in \Sigma$, constants $c>0$ and $0<\beta<1$ such that the corresponding Lyapunov function $V_F$ satisfies 
\begin{align}\label{v dot cond 2}
 \dot V_F \leq -cV_F^\beta,   
\end{align}
for all $t\in [t_{F_i},t_{F_i+1}), \; i\in \mathbb N$;
\item[(iv)] The mode $F$ is active for a cumulative duration $T_F$ defined as
\begin{align*}
T_F = \gamma (\|x_0\|) \triangleq \frac{\alpha(\|x\|)^{1-\beta}}{c(1-\beta)} + \frac{\bar \alpha(\|x\|)^{1-\beta}}{c(1-\beta)},
\end{align*} 
where $\alpha = \alpha_0+\alpha_1+\alpha_2$, $\bar \alpha = \alpha_1+\alpha_2$ and $\alpha_0 \in \mathcal{GK}$,
\end{itemize}
then the origin of \eqref{switch sys} is FTS with respect to $\sigma$. Moreover, if all the conditions hold globally and the functions $V_i$ are radially unbounded for all $i\in \Sigma$, and the functions $\alpha_1, \alpha_2\in \mathcal{GK}_\infty$, then the origin of \eqref{switch sys} is globally FTS. 
\end{Cor}

\subsection{Finite-Time Stabilizing Switching Signal}\label{sec switch signal}
In this subsection, we present a method of designing a switching signal, based upon Corollary \ref{FT Th 3}, so that the origin of the switched system is FTS. The approach is inspired from \cite{zhao2008stability} where a method of designing an asymptotically stabilizing switching signal is presented. Suppose there exist continuous functions $\mu_{ij}:\mathbb R^n\rightarrow\mathbb R$ satisfying:
\begin{equation}\label{mu ij}
    \begin{split}
        \mu_{ij}(0) & = 0, \\
        \mu_{ii}(x) & = 0 \quad \forall\; x,\\
        \mu_{ij}(x) + \mu_{jk}(x) &\leq \min\{0,\mu_{ik}(x)\}, \; \forall x
    \end{split}
\end{equation}
for all $i,j,k \in \Sigma$. Define the following sets:
\begin{equation}\label{omega i ij}
    \begin{split}
        \Omega_i &= \{x\; | \; V_i(x) -V_j(x) + \mu_{ij}(x) \leq 0, j \in \Sigma\}, \\
        \Omega_{ij} & = \{x\; |\; V_i(x) - V_j(x) + \mu_{ij}(x)  = 0, i\neq j\},
    \end{split}
\end{equation}
where $V_i$ is a Lyapunov function for each $i\in \Sigma$. 

Now we are ready to define the switching signal. Let $\sigma(t_0,x(t_0)) = i$ and $i,j\in \Sigma$ be any arbitrary modes. For all times $t\geq t_0$, define the switching signal as:
\begin{equation}\label{switch sigma}
    \begin{split}
        \sigma(t,x) & =  \left\{
              \begin{array}{lr}
                i , & \sigma(t^-,x(t^-)) = i, \;  x(t^-)\in \textrm{int} (\Omega_i); \\
                j , & \sigma(t^-,x(t^-)) = F, \;  x(t^-)\in \Omega_{Fj}, \;  \Delta_t \geq t_d; \\
                j , & \sigma(t^-,x(t^-)) = i, \;  i\neq F , \;  x(t^-)\in  \Omega_{ij}; \\
                F, & \sigma(t^-,x(t^-)) = i , \;  x(t^-)\in \Omega_{iF};
              \end{array}
           \right.  
    \end{split}
\end{equation}
where 
\begin{itemize}
    \item[-]$\Delta_t = t-t_k$ is the time duration from the last switching instant $t_k$;
    \item[-]$t_d>0$ is some positive dwell-time;
\end{itemize} 
\noindent Note that the condition for switching from mode $F$ to mode $j$ includes a dwell-time of $t_d$, so that Assumption \ref{assum: dwell time} is satisfied. We now state the following result.


\begin{Theorem}\label{FT switch law}
Let the switching signal for \eqref{switch sys} is given by \eqref{switch sigma}. Let $V_i$ are Lyapunov functions for $i = 1,2,\dots, N$, and $\mu_{ij}$ satisfy \eqref{mu ij}. Assume that the following hold:
\begin{itemize}
    \item[(i)] There exists continuous functions $\beta_{ij}:\mathbb R^n\rightarrow\mathbb R$ for $i,j\in \Sigma$ such that $\beta_{ij}(x)\leq 0$ for all $x\in \mathbb R^n$ and 
    \begin{align}\label{beta V f cond}
        \frac{\partial V_i}{\partial x}f_i(x) + \sum_{j = 1}^{N_f}\beta_{ij}(x)(V_i(x) -V_j(x) + \mu_{ij}(x))\leq 0,
    \end{align}
    holds for all $i\in \Sigma$, for all $x\in \mathbb R^n$;
    \item[(ii)] There exists a finite-time stable mode $F\in \Sigma$ satisfying condition (iii) and (iv) of Corollary \ref{FT Th 3};
    \item[(iii)] The functions $\mu_{ij}$ are continuously differentiable and satisfy
    \begin{align}\label{mu dot cond}
        \frac{\partial \mu_{ij}}{\partial x}f_i &\leq 0, \; i,j = 1, 2, \dots, N. 
    \end{align}
    \item[(iv)] No sliding mode occurs at any switching surface.
\end{itemize}
Then, the origin of \eqref{switch sys} is FTS.
\end{Theorem}
\begin{proof}
We show that all the conditions of Corollary \ref{FT Th 3} and Assumption \ref{assum: dwell time} are satisfied to establish FTS of the origin for \eqref{switch sys}, when the switching signal is defined as per \eqref{switch sigma}. As per the analysis in \cite[Theorem 3.18]{zhao2008stability}, we obtain that the conditions (i)-(ii) of Corollary \ref{FT Th 3} are satisfied with 
\begin{align}
    \alpha_1(r) & = \max_{\|x\|\leq r,\; i, j\in \Sigma_f}|\mu_{ij}(x)|,\\
    \alpha_2 (r) & = 0,
\end{align}
for any $r\geq 0$. From (ii), we obtain that conditions (iii) and (iv) of Corollary \ref{FT Th 3} hold as well. Per \eqref{switch sigma}, Assumption \ref{assum: dwell time} is also satisfied. Thus, all the conditions of the Corollary \ref{FT Th 3} and Assumption \ref{assum: dwell time} are satisfied. Hence, we obtain that the origin of \eqref{switch sys} with switching signal defined as per \eqref{switch sigma} is FTS.
\end{proof}

\begin{Remark}
Note that an arbitrary switching signal $\sigma$ may not satisfy the conditions of Corollary \ref{FT Th 3}, particularly condition (v), where the mode $F$ is required to be active for $T_F(x_0)$ time duration. For any given initial condition $x_0$, the switching signal can be defined as per \eqref{switch sigma} to render the origin of \eqref{switch sys} FTS. Definition \ref{Switch FTS Def} allows us to choose the switching signal $\sigma$ as per \eqref{switch sigma} so that the switched system \eqref{switch sys} satisfies the conditions of Corollary \ref{FT Th 3}. Moreover, one can verify that the only difference between the switching signal defined in \cite{zhao2008stability} and \eqref{switch sigma} is the introduction of dwell-time $t_d$ when switching from mode $F$. This observation re-emphasizes on the fact a system whose origin is uniformly stable can be made FTS by ensuring that the dwell-time condition and the cumulative activation time requirements are satisfied for an FTS mode.

\end{Remark}

\textbf{A note on construction of functions $\mu_{ij},V_i$}: For a class of switched systems consisting of $N-1$ linear modes and one FTS mode $F$, one can follow a design procedure similar to \cite[Remark 3.21]{zhao2008stability} to construct the functions $\mu_{ij}$, as well as the Lyapunov functions $V_i$, for all $i\neq F$. The design procedure includes choosing quadratic functions $\mu_{ij} = x^TP_{ij}x$ and $V_i = x^TR_ix$ with $R_i$ as positive definite matrices, and using the conditions \eqref{mu ij} and \eqref{mu dot cond} along with the conditions of Corollary \ref{FT Th 3}, to formulate a linear matrix inequality (LMI) based optimization problem. For system consisting of polynomial dynamics $f_i$, one can formulate a sum-of-square (SOS) problem to find polynomial functions $V_i, \mu_{ij}$ and $\beta_{ij}$ by posing \eqref{mu ij}, \eqref{beta V f cond} and \eqref{mu dot cond} inequalities as SOS constraints (see \cite{parrilo2000structured} for an overview of SOS programming and \cite{prajna2002introducing} for methods of solving SOS problems). The ``min-switching" law as described in \cite{liberzon1999basic}, can be defined by setting the functions $\mu_{ij} = 0$, which would imply that the Lyapunov functions should be non-increasing at the switching instants. Our conditions on the lines of the generalization of min-switching law, as presented in \cite{zhao2008stability}, overcome this limitation and allow the Lyapunov functions to increase at the switching instants. 

\subsection{FTS output-feedback for Switched Systems}\label{sec: linear switched FT}

In this subsection, we consider a switched linear system with $N$ modes such that only one mode is observable and controllable, and design an output-feedback to stabilize the system trajectories at the origin in a finite time. Consider the system:
\begin{equation}\label{obs switch sys dyn}
    \begin{split}
        \dot x &= A_{\sigma(t,x)} x + B_{\sigma(t,x)} u, \\
        y &= C_{\sigma(t,x)} x,
    \end{split}
\end{equation}
where $x\in \mathbb R^n, u\in \mathbb R, y\in \mathbb R$ are the system states, and input and output of the system, respectively, with $A_i\in \mathbb R^{n\times n}, B_i\in \mathbb R^{n\times 1}$ and $C_i\in \mathbb R^{1\times n}$. The switching signal $\sigma:\mathbb R_+\times\mathbb R^n\rightarrow\Sigma \triangleq \{1, 2, \dots, N\}$ is a piecewise constant, right-continuous function. We make the following assumption:

\begin{Assumption}
There exists a mode $\sigma_0\in \Sigma$ such that $(A_{\sigma_0}, B_{\sigma_0})$ is controllable and $(A_{\sigma_0}, C_{\sigma_0})$ is observable.
\end{Assumption}

Without loss of generality, one can assume that the pair $(A_{\sigma_0}, C_{\sigma_0})$ is in the controllable canonical form and $(A_{\sigma_0}, C_{\sigma_0})$ is in the observable canonical form, i.e., $A_{\sigma_0} = \begin{bmatrix}0_n & \begin{bmatrix}I_{n-1}\\0_{n-1}^T\end{bmatrix} \end{bmatrix}$, $ B_{\sigma_0} = \begin{bmatrix} 0 & 0 & 0 & \cdots & 0 & 1\end{bmatrix}^T$ and $ C_{\sigma_0} = \begin{bmatrix} 1 & 0 & 0 & \cdots & 0 & 0\end{bmatrix}$, where $I_{n-1}\in \mathbb R^{n-1\times n-1}$ is an identity matrix and $0_{k} = \begin{bmatrix} 0 & 0 & \cdots & 0 & 0\end{bmatrix}^T\in \mathbb R^{k\times 1}$. 

The objective is to design an output feedback for \eqref{obs switch sys dyn} so that the closed loop trajectories $x(\cdot)$ reach the origin in a finite time. To this end, we first design an FTS observer, and use the estimated states $\hat x$ to design the control input $u$. The form of the observer is:
\begin{equation}
    \begin{split}
        \dot {\hat x} = A_\sigma \hat x + g_\sigma(C_\sigma x-C_\sigma \hat x) + B_\sigma u.
    \end{split}
\end{equation}
Following \cite[Theorem 10]{perruquetti2008finite}, we define the function $g:\mathbb R\rightarrow\mathbb R^n$ as:
\begin{align}
    g_i(y) = l_i\sign(y)|y|^{\alpha_i}, i = 1, 2, \dots, n,
\end{align}
where $l_i$ are such that the matrix $\bar A$ defined as $\bar A = \begin{bmatrix}-\bar l & \begin{bmatrix}I_{n-1}\\0_{n-1}\end{bmatrix} \end{bmatrix} $
where $\bar l = \begin{bmatrix} l_1 & l_2 & \cdots & l_n\end{bmatrix}^T$ is Hurwitz, and the exponents $\alpha_i$ are chosen as $\alpha_i = i\alpha-(i-1)$ for $1<i\leq n$, where $1-\frac{n-1}{n}<\alpha<1$. Define the function $g_\sigma$ as:
\begin{equation}\label{g def}
    g_\sigma(y) = \left\{
              \begin{array}{cc}
                g(y), & \sigma(t) = \sigma_0;\\
                0, & \sigma(t) \neq \sigma_0;
              \end{array}
           \right.  
\end{equation}
Let the observation error be $e = x-\hat x$, with $e_i = x_i - \hat x_i$ for $i = 1, 2, \dots, N$. Its time derivative reads:
\begin{equation}\label{error obs dyn}
    \begin{split}
        \dot e = A_\sigma e-g_\sigma (C_\sigma e).
    \end{split}
\end{equation}

Next, we design a feedback $u = u(\hat x)$ so that the origin is FTS for the closed-loop trajectories of \eqref{obs switch sys dyn}. Inspired from control input defined in \cite[Proposition 8.1]{bhat2005geometric}, we define the control input as
\begin{align}\label{FTS switch control}
    u(\hat x) =  \left\{
              \begin{array}{cc}
             -\sum_{i = 1}^nk_i\sign(\hat x_i)|\hat x_i|^{\beta_i}, & \sigma(t) = \sigma_0;\\
                0, & \sigma(t) \neq \sigma_c;
              \end{array}
           \right.  ,
\end{align}
where $\beta_{j-1} = \frac{\beta_j\beta_{j+1}}{2\beta_{j+1}-\beta_j}$ with $\beta_{n+1} = 1$ and $0<\beta_n = \beta<1$, and $k_i$ are such that the polynomial $s^n+k_ns^{n-1}+\cdots +k_2s+k_1$ is Hurwitz. We now state the following result.

\begin{Theorem}
Let the switching signal $\sigma$ for \eqref{obs switch sys dyn} be given by \eqref{switch sigma} with $F = \sigma_0$. Assume that there exist functions $\mu_{ij}$ as defined in \eqref{mu ij}, and that the conditions (i)-(iii) of Theorem \ref{FT switch law} are satisfied. Then, the origin of the closed-loop system \eqref{obs switch sys dyn} under the effect of control input \eqref{FTS switch control} is an FTS equilibrium.
\end{Theorem}
\begin{proof}
We first show that there exists $T_1<\infty$ such that for all $t\geq T_1$, $\hat x(t) = x(t)$. Note that the origin is the only equilibrium of \eqref{error obs dyn}. From the analysis in Theorem \ref{FT switch law}, we know that the conditions (i) and (ii) of Corollary \ref{FT Th 3} are satisfied. 
The observation-error dynamics for mode $\sigma_0$ reads:
\begin{align}\label{obs err mode o}  
    \dot e 
    & = \begin{bmatrix} e_2- l_1\textrm{sign}(e_1)|e_1|^{\alpha_1}\\ e_3- l_2\textrm{sign}(e_1)|e_1|^{\alpha_2} \\\vdots \\
    e_n-l_{n-1} \textrm{sign}(e_1)|e_1|^{\alpha_{n-1}}\\- l_n\textrm{sign}(e_1)|e_1|^{\alpha_n} \end{bmatrix}. 
\end{align}
Now, using \cite[Theorem 10]{perruquetti2008finite}, we obtain that the origin is an FTS equilibrium for \eqref{obs err mode o}, i.e., for mode $\sigma_0$ of \eqref{error obs dyn}. From \cite[Lemma 8]{perruquetti2008finite}, we also know that \eqref{obs err mode o} is homogeneous with degree of homogeneity $d = \alpha-1<0$. Hence, using \cite[Theorem 7.2]{bhat2005geometric}, we obtain that there exists a Lyapunov function $V_o$ satisfying $\dot V_o \leq -cV_o^\beta$ where $c>0$ and $0<\beta<1$. Hence, condition (iii) of Corollary \ref{FT Th 3} is also satisfied. From the proof of Theorem \ref{FT switch law}, we obtain that the condition (iv) of Corollary \ref{FT Th 3} and Assumption \ref{assum: dwell time} are also satisfied. Hence, we obtain that the origin of \eqref{error obs dyn} is an FTS equilibrium. Thus, there exists $T_1<\infty$ such that for all $t\geq T$, $\hat x(t) = x(t)$. So, for $t\geq T_1$, the control input satisfies $u = u(\hat x) = u(x)$. Again, it is easy to verify that the origin is the only equilibrium for \eqref{obs switch sys dyn} under the effect of control input \eqref{FTS switch control}. 
The closed-loop trajectories take the following form for the mode $\sigma = \sigma_0$
\begin{align}
    \dot x 
    & = \begin{bmatrix} x_2 \\ x_3\\\vdots \\x_{n-1}\\
    x_n-\sum_{i = 1}^n k_{i} \textrm{sign}(x_i)|x_i|^{\beta_i} \end{bmatrix}. 
\end{align}
From \cite[Proposition 8.1]{bhat2005geometric}, we know that the origin of the closed-loop trajectories for mode $\sigma = \sigma_0$ is FTS. Hence, repeating same set of arguments as above, we obtain that there exists $T_2<\infty$ such that for all $t\geq T_1+T_2$, the closed-loop trajectories of \eqref{obs switch sys dyn} satisfy $x(t) = 0$. 
\end{proof}

We presented a way of designing switching signal $\sigma$ and control input $u$ for a class of switched linear system where only of the modes is controllable and observable. 
 
\section{Simulations}\label{Simulations}
We present two numerical examples to demonstrate the efficacy of the proposed methods. The first example considers an instance of the hybrid system \eqref{hybrid sys} with five modes, where one mode is FTS, one is AS, and three are unstable. We demonstrate that if the conditions of Theorem \ref{FT Th 2} are satisfied, then the trajectories of the considered system reach the origin in finite time even in the presence of unstable modes. The second example considers a switched linear control system with five modes such that only one mode is both controllable and observable. We design an FTS output controller for the considered switched system, and demonstrate that the closed-loop trajectories reach the origin despite presence of unobservable modes, and that some of the uncontrollable modes are unstable.

Note that the simulation results have been obtained by discretizing the continuous-time dynamics using Euler discretization. We use a step size of $dt = 10^{-3}$, and run the simulations till the norm of the states drops below $10^{-10}$. At this point we wish to emphasize that while the theoretical results hold for the continuous-time dynamics, and not for the implemented discretized dynamics, still the simulations reflect stable behavior that meets the theoretical bounds on the sufficiently long active time of the finite-time stable mode. In other words, we include the simulations for the sake of visualizing the theoretical results despite the discrepancy between continuous and discretized dynamics. The study of discretization methods for finite-time stable systems is left open for future investigation.  

\subsection{Example 1: Analysis of a finite-time-stable hybrid system}
We present a numerical example to illustrate the FTS results on a hybrid system given as
\begin{equation}\label{hyb sys example}
\begin{split}
    & \mathcal H = \{\mathcal F, \mathcal G , C, D\},\quad \mathcal F = \{f_1, f_2, f_3, f_4, f_5\}, \quad \mathcal G = g_1, \\
    &f_1 = \begin{bmatrix}0.01x_1^2+x_2\\-0.01x_1^3+x_2\end{bmatrix}, \quad f_2 = \begin{bmatrix}0.01x_1-x_2\\-x_1^2+0.01x_2\end{bmatrix},\\
    & f_3 = \begin{bmatrix}-x_1-x_2\\x_1-x_2\end{bmatrix},\quad 
    f_4 = \begin{bmatrix}0.01x_1^2+0.01x_1x_2\\-0.01x_1^3+x_2^2\end{bmatrix},\\
    & f_5 = \begin{bmatrix} x_2-20\textrm{sign}(x_1)|x_1|^\alpha\\-10\textrm{sign}(x_1)|x_1|^{2-2\alpha}\end{bmatrix},\\
    &g_1 = \begin{bmatrix}-1.1x_1 \\-1.1x_2\end{bmatrix},\quad C = \mathbb R^2,\quad D = \mathbb R^2,
\end{split}
\end{equation}
with $\alpha = 0.98$, where the fifth mode is FTS, and thus $F = 5$. Note that the states $x_1$ and $x_2$ change sign and increase in magnitude at the discrete jumps. The Lyapunov functions are defined as $V_i(x) = x^TP_ix$, for $i\in\{1,2,3,4\}$, with $P_1 = \begin{bmatrix}1 & 0\\ 0 & 1\end{bmatrix},  P_2 = \begin{bmatrix}5 & 2\\ 2 & 4\end{bmatrix}, P_3 = \begin{bmatrix}1 & 0\\ 0 & 3\end{bmatrix}, P_4 = \begin{bmatrix}6 & 1\\ 1 & 3\end{bmatrix}$, and $V_5(x) = \frac{k_2}{2\alpha}|x_1|^{2\alpha} + \frac{1}{2}|x_2|^2$. Note that this example is more general than the examples considered in \cite{li2019finite}, as we allow the dynamics to have unstable modes. In this example, the switches in the continuous flows occur after $0.2$ sec, i.e., $|T_{i_k}| = 0.2$ sec, and discrete jumps occur after $0.1$ sec, so that $t_d = 0.1$ sec, i.e., $|\bar T_{i_k}| = 0.1$, for all $i\in \{1, 2,\dots, 5\}$, $k\in \mathbb Z_+$ (see Assumption \ref{assum bar T tm}). 

\begin{figure}[ht!]
	\centering
	\includegraphics[width=0.9\columnwidth,clip]{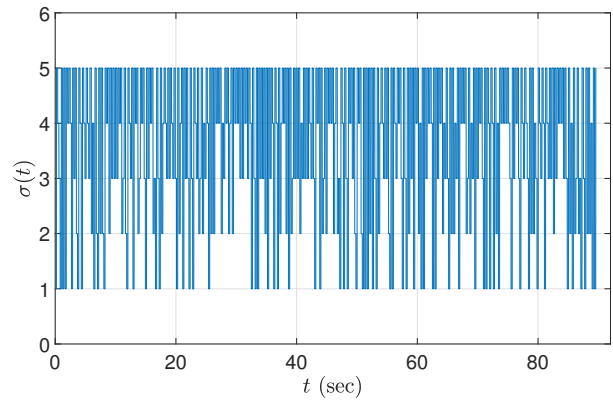}
	\caption{Switching signal $\sigma_f(t)$ for the considered hybrid system \eqref{hyb sys example}.}
	\label{fig:sim hyb mode}
\end{figure}

Figure \ref{fig:sim hyb mode} depicts the considered switching signal $\sigma_f(t)$. The switching signal is designed per Section \ref{sec switch signal} so that conditions (i) and (ii) are met; the switching signal in this example is designed using this method, but the details are omitted in the interest of space. Briefly, the Lyapunov candidates $V_i$ satisfy conditions (i) and (iii) of Theorem \ref{FT Th 2} since they are quadratic. Modes 1, 3 and 5, being stable, satisfy condition (ii) with $\alpha_2 = 0$, and modes 2 and 4, being active for a finite interval each time, satisfy condition (ii) with $\alpha_2 = k\|x_0\|^2$ for some $k>0$. It can be verified that $f_5$ is homogeneous with degree of homogeneity $d=\alpha-1<0$. Thus, using \cite[Theorem 7.2]{bhat2005geometric}, the origin is FTS under the system dynamics $f_5$, and there exists a $V_5$ satisfying \eqref{v dot cond}; therefore, condition (iv) is satisfied. Finally, the switching signal is designed so that mode 5 is active for a sufficient amount of time that satisfies condition (v).

\begin{figure}[ht!]
	\centering
	\includegraphics[width=1\columnwidth,clip]{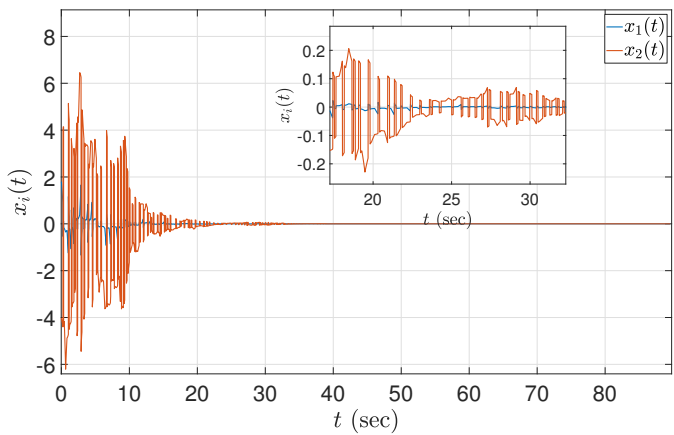}
	\caption{The evolution of $x_1(t)$ and $x_2(t)$ for hybrid system \eqref{hyb sys example}. The states can be seen switching signs during discrete jumps.}
	\label{fig:sim hyb states}
\end{figure}

Figure \ref{fig:sim hyb states} illustrates the state trajectories $x_1(t)$ and $x_2(t)$. Note that the states change sign at the discrete jumps. Figure \ref{fig:sim hyb norm} depicts the norm of the state vector $x(t)$ on log scale; note that $\|x(t)\|$ is increasing while operating in unstable modes, and decreasing while operating in stable modes. As seen in the figures, the system states, starting from $\|x(0)\| = 10$, reach to a norm of $\|x(t)\|\leq 10^{-10}$ within first 90 seconds of the simulation. Finally, Figure \ref{fig:sim hyb V} illustrates the evolution of the Lyapunov functions $V_i$ with respect to time; note that the Lyapunov functions increase, as expected, at the times of the switches in $\sigma_f$ and $\sigma_g$, as well as during the continuous flows along the unstable modes 1, 2 and 4. The provided example demonstrates that the origin of the system is FTS even when one or more modes are unstable, if the FTS mode is active for a sufficient amount of time. 

\begin{figure}[ht!]
	\centering
	\includegraphics[width=1\columnwidth,clip]{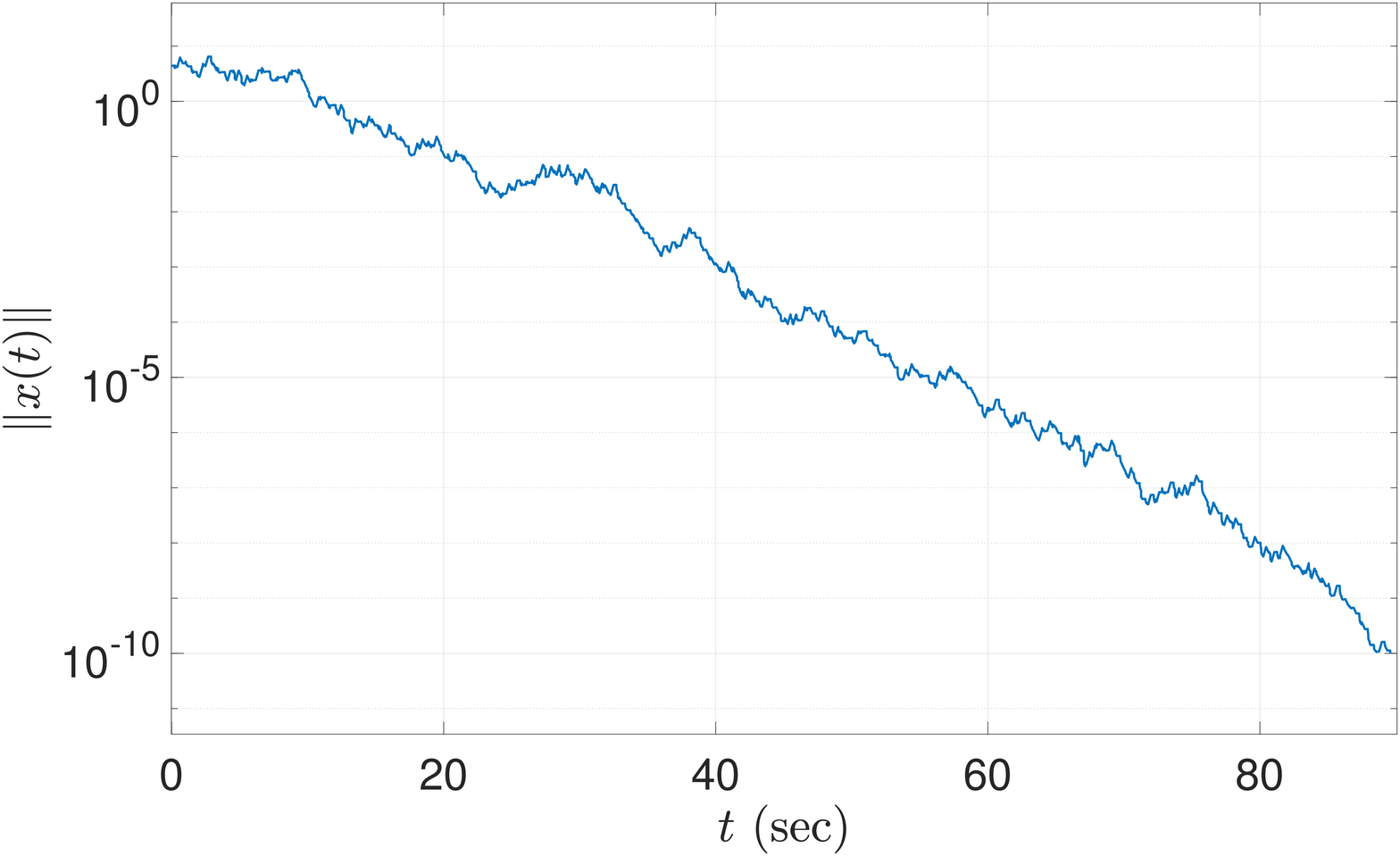}
	\caption{The evolution of $\|x(t)\|$ for \eqref{hyb sys example}. The norm of the states reach a small neighborhood of the origin within a finite time.}
	\label{fig:sim hyb norm}
\end{figure}

\begin{figure}[ht!]
	\centering
	\includegraphics[width=1\columnwidth,clip]{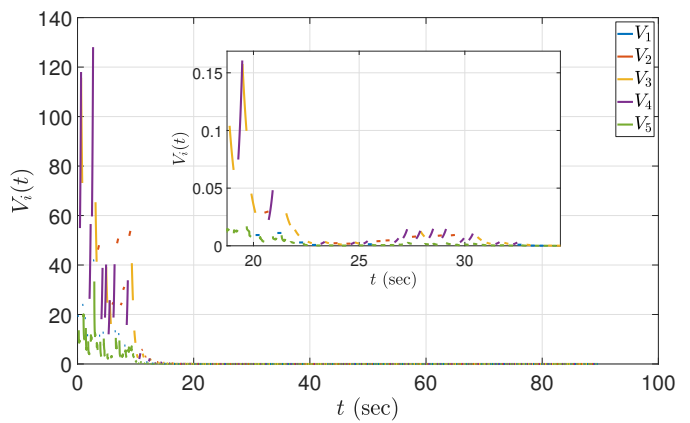}
	\caption{The evolution of the Lyapunov functions $V_i(t)$ for $t\in [0, 10]$ sec for \eqref{hyb sys example}. The Lyapunov functions for unstable modes (mode 2 and 4) increase when the respective modes are active.}
	\label{fig:sim hyb V}
\end{figure}

\subsection{Example 2: FTS output-feedback for switched linear system}
In this second example, we consider linear switched system of the form \eqref{obs switch sys dyn} and design an output feedback that stabilizes the origin for the closed-loop system in a finite time.  
For illustration purposes, we consider a system of order $n = 2$, $\sigma \in\{1,2,3,4,5\}$, and assume that mode $\sigma = 5 \triangleq \sigma_0$ is controllable and observable, i.e., that the pair $(A_{\sigma_0}, B_{\sigma_0})$ is controllable and $(A_{\sigma_0}, C_{\sigma_0})$ is observable, while other modes are either uncontrollable or unobservable, or both. The simulation parameters are:
\begin{itemize}
    \item Number of modes $N = 5$, FTS mode $F = 5$, $t_d = 0.1$, $\alpha = 0.9$, $a_1 =-10$,  $a_2 = 10$, $\beta = .9$ $k_1 = 20$ and $k_2 = 10$;
    \item The matrices $A_i, B_i, C_i$ are chosen as $A_1 = \begin{bmatrix}0 & 1\\ -1 & 0\end{bmatrix}, A_2 = \begin{bmatrix}0.1 & 0\\ 0 & 0.1\end{bmatrix}, A_3 = \begin{bmatrix}-1 & 0\\ 0 & -1.2\end{bmatrix}, A_4 = \begin{bmatrix}1 & 0.1\\ 0.1 & 2\end{bmatrix}, A_5 = \begin{bmatrix}0 & 1\\ 0 & 0\end{bmatrix}$, $B_1 =B_2 = B_3 =  B_4 = \begin{bmatrix}0\\ 0\end{bmatrix},B_5 = \begin{bmatrix}0\\ 1\end{bmatrix}$ and $C_1 = C_2 = C_3 = C_4 = \begin{bmatrix}0 & 0\end{bmatrix}, C_5 = \begin{bmatrix}1 & 0\end{bmatrix}$.
    \item Generalized Lyapunov functions are chosen as $V_i(x) = x^TP_ix$ where matrices $P_i$ are chosen as $P_1 = \begin{bmatrix}1 & 0\\ 0 & 1\end{bmatrix},  P_2 = \begin{bmatrix}5 & 2\\ 2 & 4\end{bmatrix},
        P_3 = \begin{bmatrix}1 & 0\\ 0 & 3\end{bmatrix}, P_4 = \begin{bmatrix}6 & 1\\ 2 & 3\end{bmatrix}$,
    and $V_5(x) = \frac{k_2}{2\alpha}|x_1|^{2\alpha} + \frac{1}{2}|x_2|^2$;
    \item Functions $\mu_{ij}$ as
    \begin{align*}
    \mu_{ij}(x) = \left\{
              \begin{array}{cc}
                -\|x\|^2, & i\in\{1,2,4\};\\
                0, & i\in\{3,5\};
              \end{array}
          \right.  
    \end{align*}
    for all $j\in \sigma$. 
\end{itemize}

Note that open-loop mode 1 is Lyapunov stable, mode 3 is asymptotically stable, and modes 2, 4 and 5 are unstable. The generalized Lyapunov candidates $V_i$, being quadratic, satisfy condition (i) of Corollary \ref{FT Th 3}. Modes 1, 3 and 5, being stable, satisfy condition (ii) with $\alpha_2 = 0$, and modes 2 and 4, being active only for a finite time, satisfy condition (ii) with $\alpha_2 = k\|x_0\|^2$ for some $k>0$. Conditions (iv) and (v) are satisfied by carefully designing the switching signal, as discussed in Section \ref{sec switch signal}.

\begin{figure}[ht!]
	\centering
	\includegraphics[width=1\columnwidth,clip]{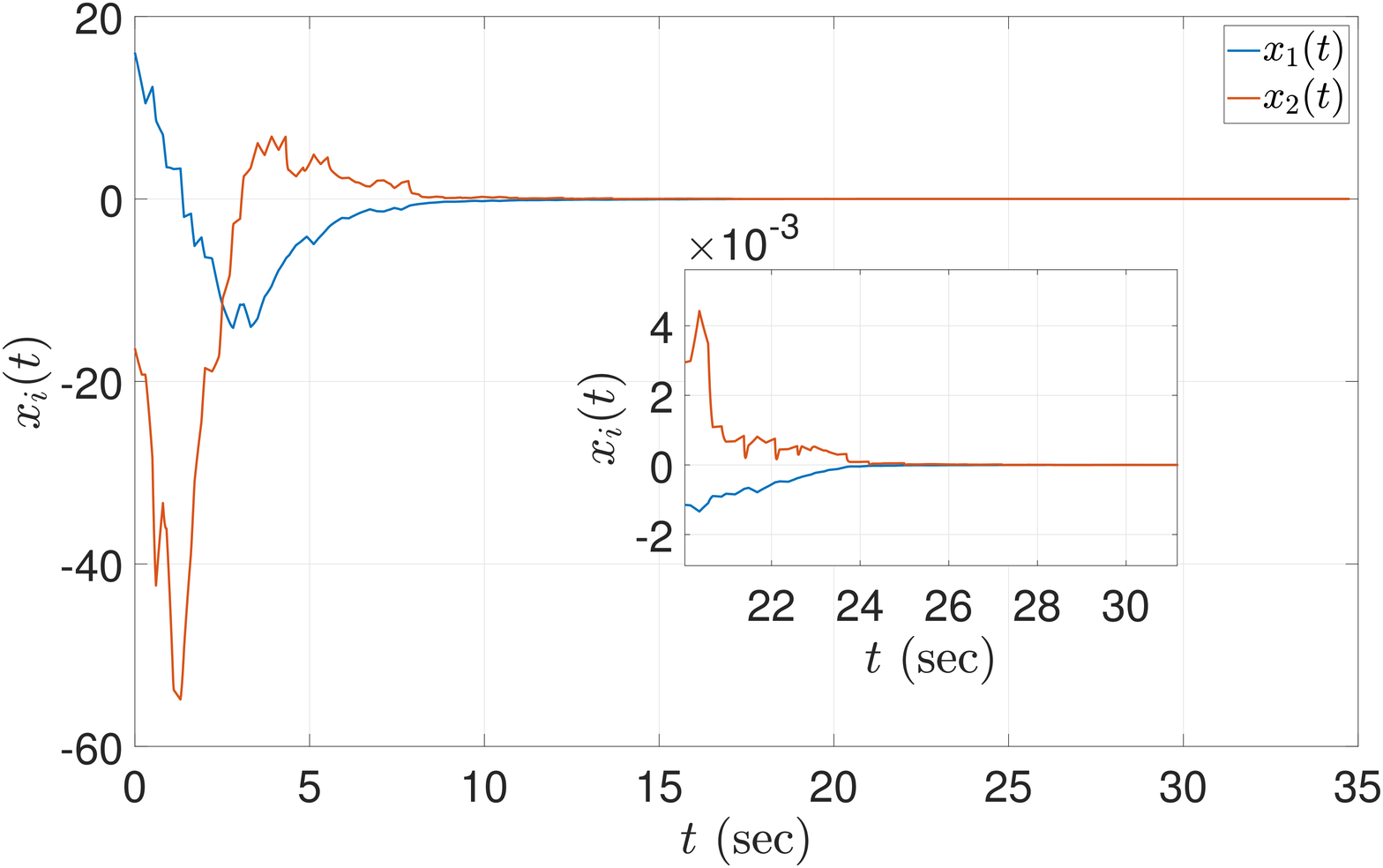}
	\caption{Closed-loop system states $x_1(t),x_2(t)$ with time for linear switched system.}
	\label{fig:sim switch x traj}
\end{figure}

Figure \ref{fig:sim switch x traj} illustrates the state trajectories $x_1(t), x_2(t)$ of the closed-loop system over time for randomly chosen initial conditions, and Figure \ref{fig:sim switch x norm} depicts the norm of the states $\|x(t)\|$. Figure \ref{fig:sim switch e norm} plots the norm of the state-estimation error, $\|x-\hat x\|$ with time. It can be seen from the these figures that both the norms $\|x\|$ and $\|x-\hat x\|$ go to zero in finite time.


\begin{figure}[ht!]
	\centering
	\includegraphics[width=1\columnwidth,clip]{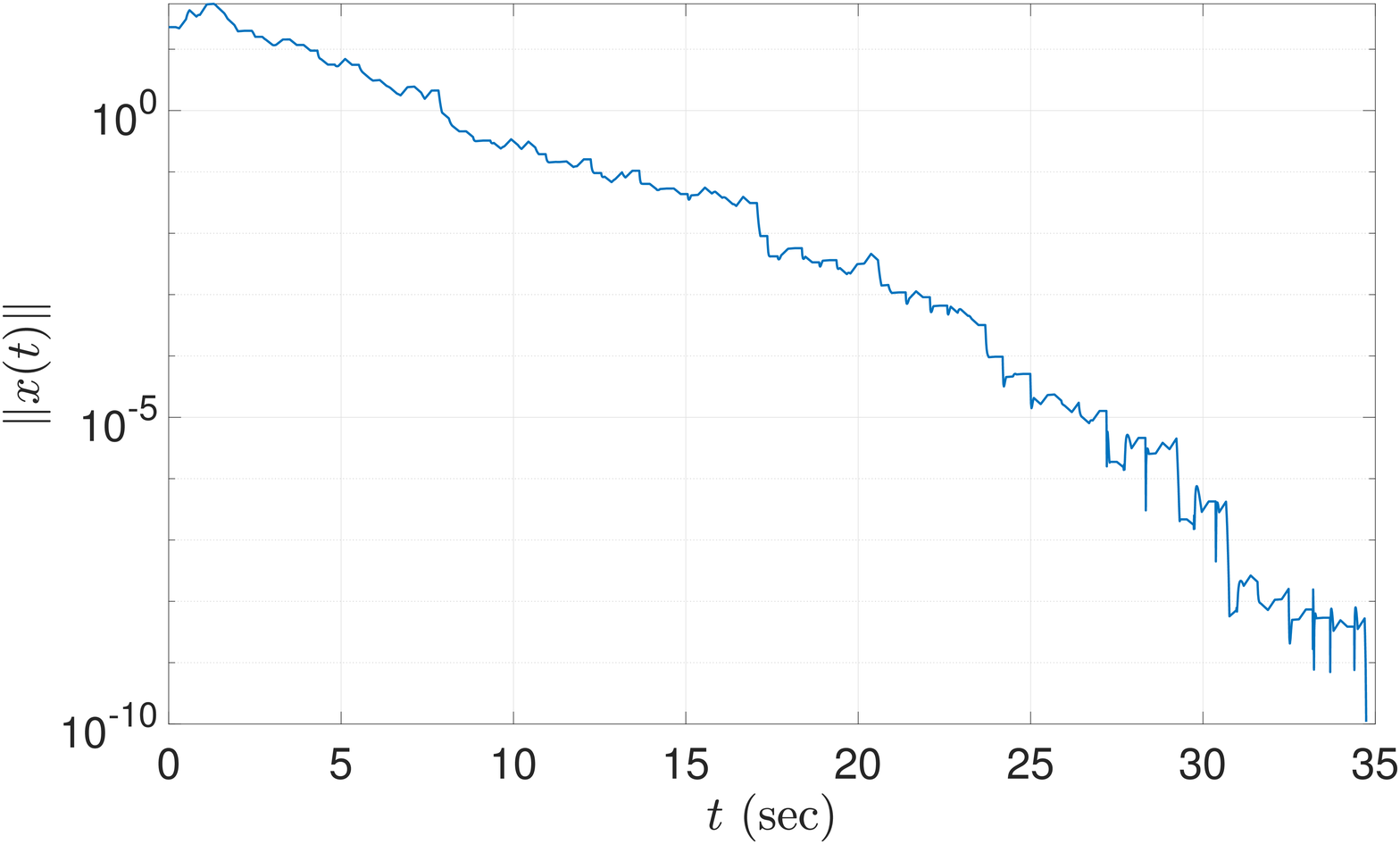}
	\caption{The norm of the state vector $x(t)$ for the closed-loop trajectories of linear switched system with time.}
	\label{fig:sim switch x norm}
\end{figure}

\begin{figure}[ht!]
	\centering
	\includegraphics[width=1\columnwidth,clip]{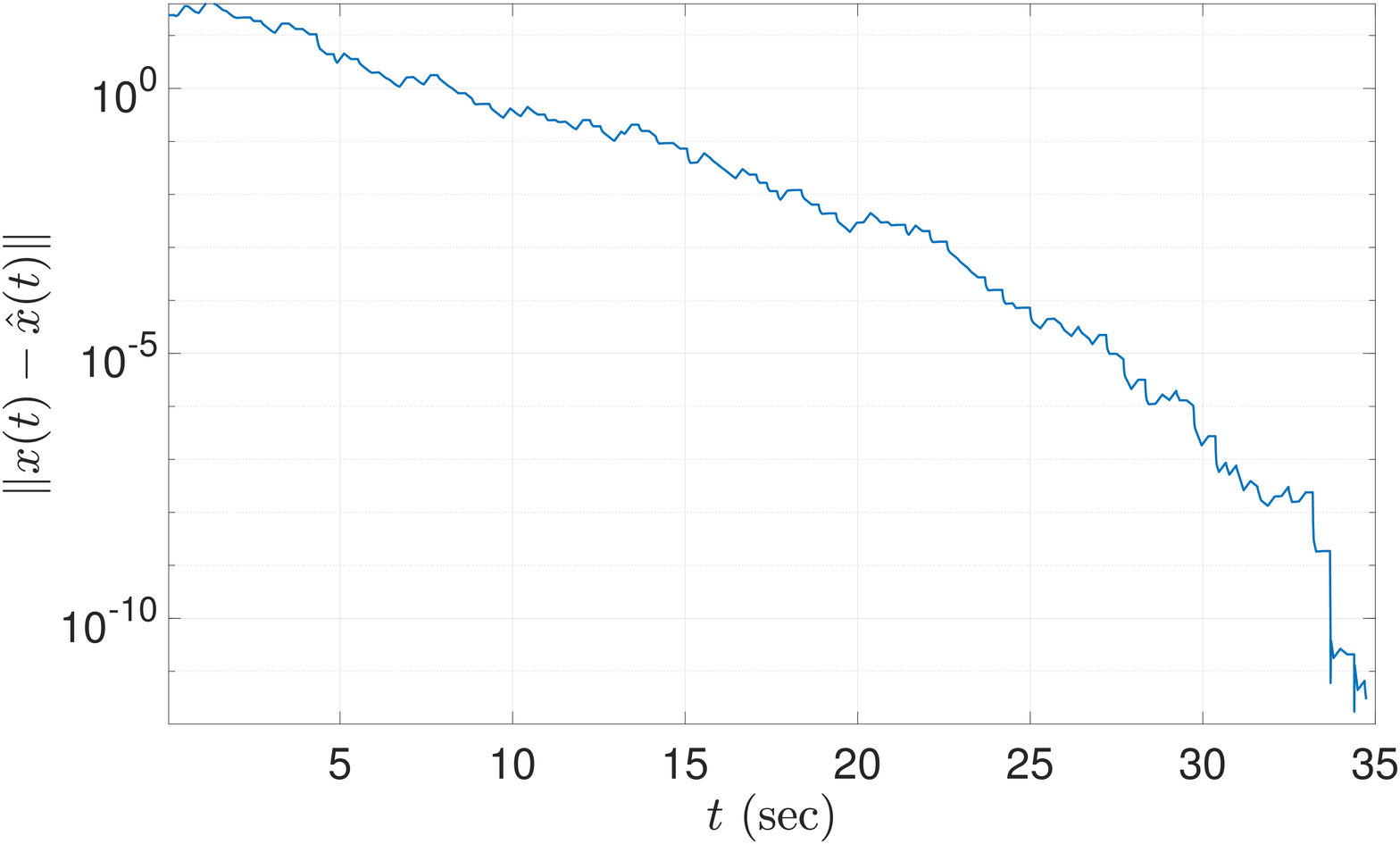}
	\caption{The norm of the state-estimation error $x(t)-\hat x(t)$ for the linear switched system with time.}
	\label{fig:sim switch e norm}
\end{figure}

Figure \ref{fig:sim switch V} shows the evolution of Lyapunov functions $V_i(x-\hat x)$ for the FTS observer of the linear switched system. It can be seen that there are unstable modes in the observer, where the value of the functions increase when the respective modes are active (e.g., mode 2 and 4). 
Finally, Figure \ref{fig:sim switch ss} plots the switching signal $\sigma$ with time. The switching signal is designed as per the design procedure listed in Section \ref{sec switch signal}. It can be seen that all the five modes (including the unstable modes) get activated for the switched linear system, while FTS of the origin is still ensured.

\begin{figure}[ht!]
	\centering
	\includegraphics[width=1\columnwidth,clip]{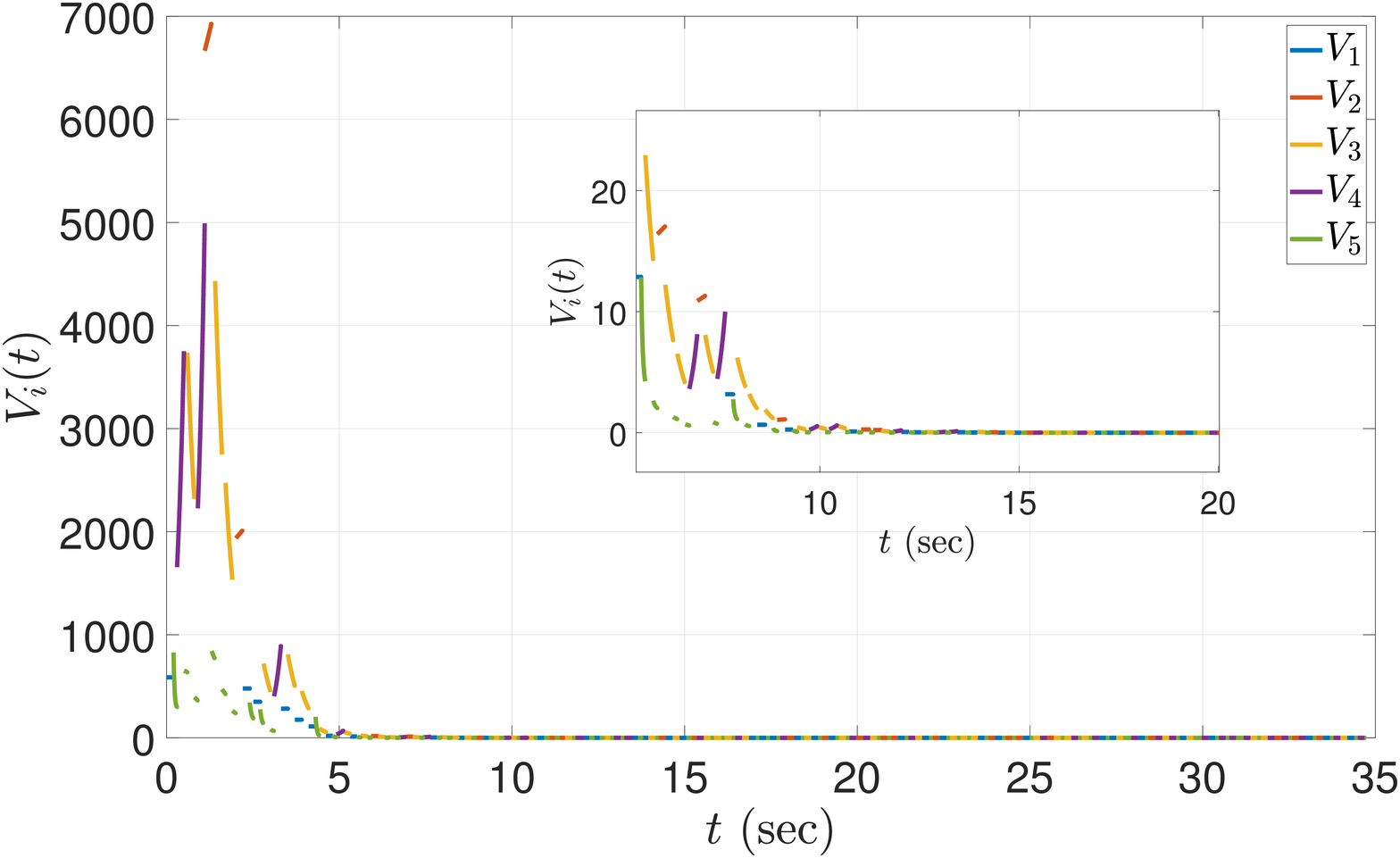}
	\caption{The evolution of the Lyapunov functions $V_i(t)$ for the FTS observer of the linear switched system.}
	\label{fig:sim switch V}
\end{figure}

\begin{figure}[ht!]
	\centering
	\includegraphics[width=1\columnwidth,clip]{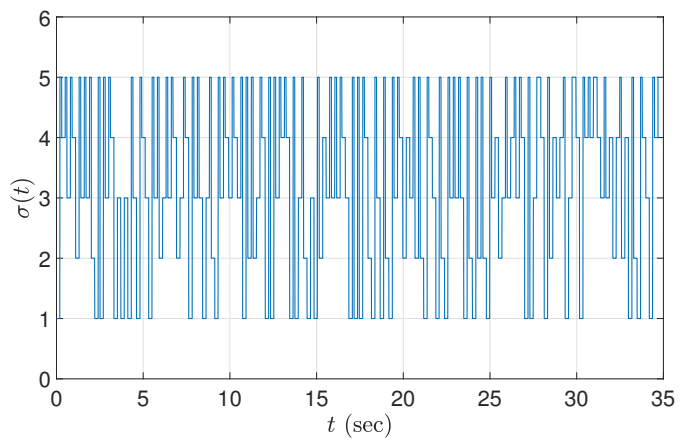}
	\caption{Switching signal for the linear switched system.}
	\label{fig:sim switch ss}
\end{figure}

The provided examples validate that the system can achieve FTS even when one or more modes are unstable, if the FTS mode is active for long enough.




\section{Conclusions and Future Work}\label{Conclusions}
In this paper, we studied FTS of a class of switched and hybrid systems. We showed that under some mild conditions on the bounds on the difference of the values of Lyapunov functions, if the FTS mode is active for a sufficient cumulative time, then the origin of the hybrid system is FTS. Our proposed method allows the individual Lyapunov functions to increase both during the continuous flows as well as at the discrete state jumps, i.e., it allows the hybrid system to have unstable modes. We also presented a method of designing a finite-time stabilizing switching signal. As an application of the theoretical results, we designed an FTS output feedback for a class of linear switched systems in which only one of the modes is both controllable and observable.  


Control inputs that satisfy the multiple-Lyapunov-function conditions are typically obtained via optimization-based techniques, for example linear matrix inequalities (LMIs) for linear switched systems, or sum-of-squares (SOS) for switched systems of polynomial dynamics. In addition, state and time constraints can be further imposed to the underlying optimization problems to capture spatiotemporal specifications. Our ongoing research focuses on incorporating input and state constraints in the hybrid systems framework to model safety (in the sense of invariance of a safe set of states) and temporal requirements (in the sense of convergence to a set or to a point within an arbitrarily chosen time, if possible). More specifically, we are investigating how to impose convergence of the system trajectories in a prescribed time that can be \emph{a priori} selected by the user, rather than merely in finite time (which depends on the initial conditions, and hence can not be in general chosen arbitrarily), so that the overall framework can be used for the synthesis and analysis of controllers for spatiotemporal specifications.


\bibliographystyle{IEEEtran}
\bibliography{myreferences}

\appendices

\section{Proof of Lemma 1}
\begin{proof}\label{app lemma 1 pf}
Lemma 3.3, 3.4 of \cite{zuo2016distributed} establish the following set of inequalities for $z_i\geq 0$ and $0<r\leq 1$
\begin{align}
    \left(\sum_{i = 1}^M z_i \right)^r \leq \sum_{i = 1}^M z_i^r \leq M^{1-r} \left(\sum_{i = 1}^M z_i \right)^r.
\end{align}
Hence, we have that for $a\geq b\geq 0$ and $0<r\leq 1$, $ a^r = (b + (a-b))^r \leq b^r  + (a-b)^r$, or equivalently, 
\begin{equation}\begin{split}
    a^r-b^r \leq (a-b)^r.
\end{split}\end{equation}
Hence, we have that for any $0<r\leq 1$, 
\begin{align*}
    \sum_{i = 1}^k (a_i^r-b_i^r) \leq \sum_{i\in I_1}(a_i^r-b_i^r) & \leq \sum_{i\in I_1}(a_i-b_i)^r.
\end{align*}



\end{proof}

\section{Proof of Corollary \ref{cor LS FTS mode}}\label{app:proof cor LS FTS mode}
\begin{proof}
Since the origin is uniformly stable, we know that there exist $\alpha_4\in \mathcal{GK}$ and a constant $c>0$ such that 
\begin{align}\label{LS xt ineq}
    \|x(t)\|\leq \alpha_4(\|x_0\|), 
\end{align}
for all $t\geq 0$ and all $\|x_0\|<c$ (\cite[Lemma 4.5]{khalil1996noninear}). Now, since the function $V_F$ is positive definite, we know that there exists $\alpha_5\in \mathcal{GK}$ such that (\cite[Lemma 4.3]{khalil1996noninear})
\begin{align*}
    V(x(t))\leq \alpha_5(\|x(t)\|)\overset{\eqref{LS xt ineq}}{\leq}\alpha_5(\alpha_4(\|x_0\|))  = \alpha(\|x_0\|),
\end{align*}
where $\alpha = \alpha_5\circ \alpha_4\in \mathcal{GK}$. Using this, we obtain that 
\begin{align*}
   \bar T_F = \sum\limits_{k = 1}^M |\bar T_{F_k}| & \leq \sum\limits_{k = 1}^M\Big(\frac{\bar V_{F_k}^{1-\beta}}{c((1-\beta)}-\frac{\bar V_{F_k+1}^{1-\beta}}{c(1-\beta)}\Big)\\
   & \leq \sum\limits_{k = 1}^M\frac{\bar V_{F_k}^{1-\beta}}{c((1-\beta)}\leq \sum\limits_{k = 1}^M\frac{\alpha(\|x_0\|)^{1-\beta}}{c((1-\beta)}.
\end{align*}
Define $\bar \gamma = \sum\limits_{k = 1}^M\frac{\alpha^{1-\beta}}{c((1-\beta)}\in \mathcal{GK}$ to be complete the proof. 
\end{proof}


\end{document}